\pdfoutput=1
\documentclass[a4paper,reqno]{amsart}

\usepackage{mathtools}
\usepackage[american]{babel}

\usepackage{microtype}

\usepackage[dvipsnames]{xcolor}
\usepackage{hyperref}
\hypersetup{
    colorlinks=true,
    citecolor=Green,
    linkcolor=NavyBlue,
    filecolor=magenta,
    urlcolor=BrickRed
    }

\usepackage{enumitem}
\usepackage{fontawesome5} 
\usepackage{tikz-cd}
\usepackage[capitalize]{cleveref}

\usepackage[backend=biber,%
minalphanames=3,%
maxalphanames=4,%
maxnames=5,%
style=alphabetic,%
useprefix=true]{biblatex}
\usetikzlibrary{decorations.markings,decorations.pathreplacing,
  shapes.geometric,matrix,arrows,chains,positioning,scopes}

\pgfdeclarearrow{
  name = pxto,
  setup code = {
    \pgfarrowssettipend{1.5\pgflinewidth}
    \pgfarrowssetbackend{-2.5508\pgflinewidth}
    \pgfarrowssetlineend{-.25\pgflinewidth}
    \pgfarrowssetvisualbackend{-0.021\pgflinewidth}
    \pgfarrowsupperhullpoint{1.5\pgflinewidth}{0\pgflinewidth}
    \pgfarrowsupperhullpoint{-2.0085\pgflinewidth}{3.6525\pgflinewidth}
    \pgfarrowsupperhullpoint{-2.5508\pgflinewidth}{3.0763\pgflinewidth}
  },
  drawing code = {
    \pgfsetdash{}{0pt}%
    \pgfpathmoveto{\pgfpoint{1.5\pgflinewidth}{0.0254\pgflinewidth}}%
    \pgfpathlineto{\pgfpoint{-2.0085\pgflinewidth}{3.6525\pgflinewidth}}%
    \pgfpathlineto{\pgfpoint{-2.5508\pgflinewidth}{3.0763\pgflinewidth}}%
    \pgfpathlineto{\pgfpoint{-0.4322\pgflinewidth}{0.5\pgflinewidth}}%
    \pgfpathlineto{\pgfpoint{-0.4322\pgflinewidth}{-0.5\pgflinewidth}}%
    \pgfpathlineto{\pgfpoint{-2.5508\pgflinewidth}{-3.0763\pgflinewidth}}%
    \pgfpathlineto{\pgfpoint{-2.0085\pgflinewidth}{-3.6525\pgflinewidth}}%
    \pgfpathclose%
    \pgfusepathqfill
  }
}
\tikzset{>=pxto}
\tikzcdset{arrow style=tikz}

\newcommand*{\pbcorner}[1][dr]{\ar[#1,phantom,"\lrcorner" , very near start]}
\newcommand*{\pocorner}[1][dr]{\ar[#1,phantom,"\ulcorner" , very near end]}

\DeclarePairedDelimiter{\pa}{(}{)}
\newcommand{\colonequiv}{\mathrel{\vcentcolon\mspace{-1mu}\equiv}}

\DeclarePairedDelimiter{\squash}{\|}{\|}
\DeclarePairedDelimiter{\tosquash}{|}{|}
\DeclarePairedDelimiter\Trunc{\lVert}{\rVert} 
\DeclarePairedDelimiter{\Bgen}{\B\langle}{\rangle}
\DeclarePairedDelimiter{\gen}{\langle}{\rangle}

\DeclareMathOperator{\id}{id}

\DeclareMathOperator{\inl}{inl}
\DeclareMathOperator{\inr}{inr}
\DeclareMathOperator{\refl}{refl}
\DeclareMathOperator{\ap}{ap}
\DeclareMathOperator{\ev}{ev}
\DeclareMathOperator{\fib}{fib}
\DeclareMathOperator{\Hom}{Hom}
\newcommand*{\B}{\mathrm{B}}
\newcommand*{\F}{\mathrm{F}}
\newcommand*{\BF}{\B\F}
\newcommand*{\BA}{\mathrm{BA}}
\newcommand*{\BH}{\mathrm{BH}}
\newcommand*{\AG}{\mathrm{A}} 
\newcommand*{\HG}{\mathrm{H}} 
\DeclareMathOperator{\Aut}{Aut}
\DeclareMathOperator{\north}{N}
\DeclareMathOperator{\south}{S}
\DeclareMathOperator{\merid}{m}
\newcommand*{\pt}{{\mathrm{pt}}}
\DeclareMathOperator{\const}{const}
\DeclareMathOperator{\im}{im}

\DeclareMathOperator{\susp}{\Sigma}

\newcommand*{\ptdto}{\to_\pt}%
\newcommand*{\join}{\mathbin{\ast}}
\newcommand*{\smashpr}{\mathbin{\wedge}}
\newcommand*{\wedgesum}{\mathbin{\vee}}
\newcommand*{\popr}{\mathbin{\square}}
\newcommand*{\Zero}{\mathbf{0}}
\newcommand*{\One}{\mathbf{1}}
\newcommand*{\Two}{\mathbf{2}}
\newcommand*{\Z}{\mathbb Z}
\newcommand*{\Circle}{S^1}
\newcommand*{\Sphere}[1]{S^{#1}}
\newcommand*{\loopc}{\mathrm{loop}}

\newcommand*{\Set}{\mathrm{Set}}
\newcommand*{\upH}{\mathrm{H}}
\newcommand*{\upK}{\mathrm{K}}

\newcommand*{\HatcherStr}{\textrm{Hatcher-structure}}

\DeclareMathOperator{\basept}{pt}

\DeclareMathOperator{\colim}{colim}

\DeclareMathOperator{\Grp}{\mathsf{Grp}}

\newcommand{\formalized}{{\color{NavyBlue!75!White}{\raisebox{-0.5pt}{\scalebox{0.8}{\faCog}}}}}
\newcommand{\flinkurl}[1]{\href{#1}{\formalized}}

\newcommand{\flinkspec}[3]{\flinkurl{https://unimath.github.io/agda-unimath/#1.#2.html\##3}}

\theoremstyle{plain}
\newtheorem{theorem}{Theorem}[section]
\newtheorem{lemma}[theorem]{Lemma}
\newtheorem{proposition}[theorem]{Proposition}
\newtheorem{corollary}[theorem]{Corollary}
\theoremstyle{definition}
\newtheorem{definition}[theorem]{Definition}

\newtheorem{principle}[theorem]{Principle}
\theoremstyle{remark}
\newtheorem{remark}[theorem]{Remark}

\begin{document}

\title[Epimorphisms and Acyclic Types in Univalent Foundations]%
{Epimorphisms and Acyclic Types \\ in Univalent Foundations}

\author[Buchholtz]{Ulrik Buchholtz}
\author[de Jong]{Tom de Jong}
\address{University of Nottingham, UK}
\email{{\href{mailto:ulrik.buchholtz@nottingham.ac.uk}{\texttt{ulrik.buchholtz@nottingham.ac.uk}} \\
    {\href{mailto:tom.dejong@nottingham.ac.uk}{\texttt{tom.dejong@nottingham.ac.uk}}}}}
\urladdr{\url{https://ulrikbuchholtz.dk/} \\ \url{https://tdejong.com}}

\author[Rijke]{Egbert Rijke}
\address{University of Ljubljana, Slovenia and Johns Hopkins University, MD, USA}
\email{\href{mailto:erijke1@jhu.edu}{erijke1@jhu.edu}}
\urladdr{\url{https://egbertrijke.github.io/}}

\begin{abstract}
  We characterize the epimorphisms in homotopy type theory (HoTT) as the
  fiberwise acyclic maps and develop a type-theoretic treatment of acyclic maps
  and types in the context of synthetic homotopy theory
  as developed in univalent foundations.
  We present examples and applications in group theory, such as the acyclicity
  of the Higman group, through the identification of groups with 0-connected,
  pointed 1-types.
  Many of our results are formalized as part of the agda-unimath library.
\end{abstract}


\maketitle

\section{Introduction}

Univalent Foundations relies on a homotopical refinement of the
\emph{propositions as types} approach to logical reasoning in dependent type
theory, thence also known as homotopy type theory (HoTT)~\cite[Ch.~3]{HoTTBook}.
One virtue of HoTT is that many advanced concepts from homotopy theory can be
expressed in simple logical terms, sidestepping encodings in terms of
combinatorial or point-set topological notions of spaces.  The corresponding
program is known as \emph{synthetic homotopy
  theory}~\cite{Awodey2012,Buchholtz2019,Shulman2021}.  Additional benefits are
that most results can be developed in a basic system of very modest
proof-theoretic strength~\cite{Rathjen2018}, way below that of classical
second-order arithmetic, and that the results apply more generally than
classical homotopy theory, namely in any higher topos~\cite{Shulman2019}.  Here,
we consider the notion of epimorphism of types in HoTT---%
and its deep connections to synthetic homotopy theory---%
paying close attention to the logical principles needed throughout.

A map \(f :A \to B\) is an \emph{epimorphism} if it has the desirable property
that for any map \(f' : A \to X\), there is at most one extension (dashed in the
diagram below) of \(f'\) along \(f\).
\begin{equation}\label{ext-diagram}
  \begin{tikzcd}
    A \ar[d,"f"']\ar[r,"f'"] & X \\
    B \ar[ur,dashed]
  \end{tikzcd}
\end{equation}
In (\(1\)-)category theory, this property is often equivalently phrased as: for any
two maps \(g,h : B \to X\), if \(g \circ f = h \circ f\), then \(g = h\).
It is well known that a map between sets is an epimorphism precisely when it is
surjective.
In HoTT one also considers \emph{higher
  types} that don't necessarily behave as sets, because in general, equality
types can have non-trivial structure.
As a consequence, the notion of epimorphism in HoTT becomes more involved and
rather interesting. We shall illustrate this with an example.

\subsection*{Epimorphisms and the circle}

To see that something unusual is going on in the presence of higher types, we
will show that, while the terminal map \(\Two\to\One\) is an epimorphism of
sets, it is \emph{not} an epimorphism of (higher) types. %
In fact, we claim that the type of extensions of \(\Two\to\Circle\) along
\(\Two\to\One\) is equivalent to \(\mathbb Z\).

Recall that the circle \(\Circle\) is the higher inductive type
generated by a base point \(\basept : \Circle\) and an identification
\({\loopc : \basept = \basept}\).
A standard result~\cite[Sec.~8.1]{HoTTBook} is that the loop space
\(\basept = \basept\) is equivalent to the type of integers \(\mathbb Z\).

Now consider diag.~\eqref{ext-diagram} where \(f\) is the map \(\Two\to\One\)
and \(f'\) is the constant map \(\Two\to\Circle\) pointing at the base point.
The type of extensions of \(f'\) along \(f\) is then equivalent to
\( \sum_{x : \Circle}\pa*{x = \basept} \times \pa*{x = \basept}\). Since
\(\sum_{x:\Circle}\pa*{x=\basept}\) is contractible it follows that the type of
extensions of \(f'\) along \(f\) is equivalent to the loop space
\(\basept =\basept\), which is in turn equivalent to~\(\mathbb Z\). Therefore we
see that the type of extensions of \(g\) along \(f\) has infinitely many
elements.  Thus, in a suitable, higher sense, the map \(\Two \to \One\) is
\emph{not} an epimorphism of (higher) types.

\subsection*{Related work}

The first main result of our paper is the characterization of epimorphisms in
homotopy type theory (HoTT) as those maps whose fibers are all
\emph{acyclic}. This result is expected from classical results in algebraic
topology~\cite{HausmannHusemoller1979,Alonso1983} and higher topos
theory~\cite{Hoyois2019,Raptis2019}, but new in HoTT.
Traditionally, acyclic spaces are important in K-theory~\cite{Kbook} and a space
is defined to be acyclic if its reduced integral homology vanishes. We instead
define a type to be acyclic if its \emph{suspension} is contractible and we relate
these two definitions in \cref{sec:homology}.

One can understand our results purely type-theoretically, but, at the same time,
our results apply to all Grothendieck \((\infty,1)\)-topoi---and not just the
\(\infty\)-topos of spaces---since HoTT can be seen as an internal language of
higher topoi~\cite{Shulman2019}.
This highlights an important difference between our work and that of
Raptis~\cite{Raptis2019} and Hoyois~\cite{Hoyois2019}.
The former applies to the \(\infty\)-topos of spaces only and sometimes relies
on tools only available there such as Whitehead's
Principle~\cite[Sec.~8.8]{HoTTBook}.
Hoyois establishes the acyclic maps as the left class of a modality on an
arbitrary \(\infty\)-topos, but uses site presentations to do so. In contrast,
the arguments of our work are fully internal.
The closure results of this paper (\cref{sec:acyclic-closure-properties})
along with the Blakers--Massey theorem and its dual (\cref{sec:blakers-massey}) would
follow if we could construct the modality of acyclic maps in HoTT, but since we
don't (yet) know how to do this, we offer direct proofs instead.
In the case of Blakers--Massey, Raptis gives a non-constructive argument, while
we offer a more constructive account, but still relying on an axiom.
In algebraic topology, such direct proofs were given for acyclic maps between
CW-spaces (spaces having the homotopy type of a CW-complex) by Hausmann and
Husemoller~\cite[Sec.~2]{HausmannHusemoller1979}; and Alonso~\cite[Sec.~4]{Alonso1983}
further studied acyclic maps between (path-connected) CW-spaces.
In a few places (i.e.,
\cref{fiber-is-cofiber,acyclic-extension,acyclic-iff-balanced}) we give direct
references to analogous results in
\cite{HausmannHusemoller1979,Alonso1983,Raptis2019} for comparison.

We emphasize that the proofs of the nontriviality and acyclicity of the examples
in \cref{sec:examples-acyclic} are new. For the first example,
Hatcher~\cite[Ex.~2.38]{HatcherAT} proves acyclicity of the complex using a
calculation in homology, whereas we derive acyclicity fairly directly from
Eckmann--Hilton~\cite[Thm.~2.1.6]{HoTTBook}.
The second example is the classifying type of the Higman group~\cite{Higman1951}
and the originality of our proof that this group is nontrivial is further
commented on below and in \cref{sec:higman}.

Finally, we mention that all our proofs are fully constructive and do not rely
on classical principles such as the axiom of choice and excluded
middle. We also do not make use of impredicativity in the form of
propositional resizing~\cite[Sec.~3.5]{HoTTBook}.

\subsection{Outline}

\begin{itemize}
\item Our first main result is \cref{acyclic-characterization} which
  characterizes the epimorphisms in homotopy type theory as the fiberwise
  acyclic maps.
  In \cref{sec:epis-and-acyclic-maps} we further prove closure properties and
  show that the epimorphisms are also exactly the balanced maps of
  Raptis~\cite{Raptis2019}.
\item \cref{sec:k-acyclic} introduces relativized notions of acyclicity and
  epimorphisms to \(k\)\nobreakdash-types (\cref{characterization-of-k-epis})
  with applications in group theory via the delooping of a group as a pointed
  connected \(1\)-type.
\item Although we leave establishing the acyclic maps as the left class of an
  orthogonal factorization system to future work, we do prove Blakers--Massey
  for acyclic maps (\cref{BM-acyclic}), and we study what the corresponding right
  class should be: the hypoabelian maps (\cref{sec:hypoab}).
  These results depend on an additional axiom that we call the \emph{plus
    principle} (\cref{sec:plus}).
\item In \cref{sec:homology} we relate our definition of an acyclic type to the
  classical definition which requires its reduced integral homology to be
  zero.
  We also discuss the situation for maps, which is a bit more subtle.
\item Finally, we exhibit interesting examples of acyclic types
  (\cref{sec:examples-acyclic}).
  \begin{itemize}
  \item The first example is a type-theoretic incarnation of Hatcher's
    two-dimensional complex~\cite[Ex.~2.38]{HatcherAT} and positively answer a
    question raised by Rezk~\cite[p.~11]{Rezk2019}.
  \item The second example is the classifying type of the Higman
    group~\cite{Higman1951}.
    In proving its nontriviality we make use of recent results by
    David W\"arn~\cite{Warn2023}.
    It is noteworthy
    that our proof is constructive and avoids
    combinatorial group theory~\cite{LyndonSchupp2001}, relying on higher
    categorical tools instead.
  \end{itemize}
\end{itemize}

\subsection{Foundations and preliminaries}
We work in homotopy type theory (HoTT), also known as univalent foundations, and
employ the conventions and notations of the HoTT book~\cite{HoTTBook}. We make
use of the univalence axiom and its consequences (such as function
extensionality) without mentioning this explicitly.

We also assume familiarity with \emph{pullbacks}~\cite[Exer.~2.11]{HoTTBook},
\emph{higher inductive types}~\cite[Sec.~6]{HoTTBook} and specifically,
\emph{pushouts}~\cite[Sec.~6.8]{HoTTBook}.
We also assume familiarity with \(k\)-types~\cite[Sec.~7.1]{HoTTBook} and recall
that a type \(A\) is \emph{\(k\)-connected}~\cite[Sec.~7.5]{HoTTBook} if its
\emph{\(k\)-truncation}~\cite[Sec.~6.9]{HoTTBook} is contractible, i.e.\
\(\squash{A}_k \simeq \One\), and that this notion extends to maps by
considering \emph{fibers}~\cite[Def~4.2.4]{HoTTBook}.

A recurring idea, developed in~\cite{BDR2018,Symmetry}, is to regard a group
\(G\) via its \emph{classifying type} \(\B G\): this is a \(0\)-connected,
pointed \(1\)-type such that taking loops at the point recovers the group \(G\),
i.e., we have an isomorphism of groups \(G \cong (\pt =_{\B G} \pt)\).
(The latter is indeed a group as we can compose and invert loops and these
operations are neutral on the trivial constant loop.)

\subsection{Formalization}

We used the Agda proof assistant~\cite{Agda} to formalize substantial parts of
this paper in the agda-unimath library~\cite{agda-unimath}.
Where appropriate, definitions, lemmas, theorems, etc.\ are marked with the
symbol \(\formalized\) that is a link to (the HTML rendering of) the relevant
file in the agda-unimath library.

\section{Epimorphisms and acyclic maps}\label{sec:epis-and-acyclic-maps}

In the theory of (1-)categories, a map \(f : A \to B\) is an epimorphism if any
two maps \(g,h : B \to X\) are equal as soon as \({g \circ f = h \circ f}\). In
other words, \(f\) is an epimorphism if precomposition with \(f\) is injective.
To get a homotopically well~behaved notion, i.e., to ensure that being epic is a
\emph{property} of a map, we replace the notion of injection by the
notion of embedding in our definition of
epimorphism. Recall~\cite[Def.~4.6.1]{HoTTBook} that a map is an embedding if
its action on identity types is an equivalence.
For sets the notions of embedding and injection coincide, so that
we recover the original notion of epimorphism in such cases.

\begin{definition}[\flinkspec{foundation}{epimorphisms}{definitions} %
  Epimorphism]\label{def:epic}
  A map \({f : A \to B}\) is an \emph{epimorphism} if for every type \(X\), the
  precomposition map
  \[
    (B \to X) \xrightarrow{f^\ast} (A \to X)
  \]
  is an embedding. We also say that \(f\) is an \emph{epi} or that it is
  \emph{epic}.
\end{definition}

Thus, \(f\) is epic precisely when the canonical function
\(g = h \to g \circ f = h \circ f\) is an equivalence for all maps
\(g,h : B \to X\).

\begin{remark}\label{extensions-as-fiber}
  Notice that the fiber of \(f^\ast\) at a map \(g : B \to X\) is exactly the
  type of extensions of \(g\) along \(f\), as considered at the start of the
  introduction to this paper.
  Since embeddings can be characterized as those maps whose fibers are all
  propositions~\cite[Thm.~12.2.3]{Rijke2022}, we see that \(f\) is an
  epimorphism if and only if for all \(g : B \to X\) the type of extensions of
  \(g\) along \(f\) is a proposition, i.e.\ every such \(g\) has at most one
  extension along \(f\).
\end{remark}

The universal quantification over types \(X\) in \cref{def:epic} should be
understood as ranging over types \(X\) in a type universe, so a priori this
definition only makes sense relative to a universe. However, it is consequence
of the characterization of epimorphisms (\cref{acyclic-characterization}) that
the notion is actually independent of the type universe.

\begin{lemma}
  [\flinkspec{foundation}{epimorphisms}{the-codiagonal-of-an-epimorphism-is-an-equivalence}]%
  \label{epic-iff-pushout-square}
  A map \(f : A \to B\) is epic if and only if the commutative square
  \[
    \begin{tikzcd}
      A \ar[r,"f"] \ar[d,"f"'] & B \ar[d,"\id"] \\
      B \ar[r,"\id"'] & B
    \end{tikzcd}
  \]
  is a pushout.
\end{lemma}
\begin{proof}
  The square is a pushout if and only if the map
  \begin{align*}
    (B \to X) &\to \sum_{g : B \to X}\sum_{h : B \to X}\pa*{g \circ f = h \circ f} \\
    k &\mapsto (k,k,\refl)
  \end{align*}
  is an equivalence for every type \(X\).
  Note that the right hand side is equivalent to
  \(\sum_{g : B \to X}\fib_{(-) \circ f}(g \circ f)\), so that we have this
  equivalence exactly when \(\fib_{(-) \circ f}(g \circ f)\) is contractible for
  all types \(X\) and maps \(g : B \to X\).
  But this happens if and only if \((-) \circ f\) is an embedding, i.e., when \(f\) is epic.
\end{proof}

Since we work in dependent type theory, it is natural to also
consider the following, seemingly stronger notion:

\begin{definition}
  [\flinkspec{foundation}{dependent-epimorphisms}{definitions} Dependent epimorphism]
  A \emph{dependent epimorphism} is a map \(f : A \to B\) such that for every
  type family \(P\) over \(B\), the precomposition map
  \[
    \prod_{b : B}P(b) \xrightarrow{f^\ast} \prod_{a : A}P(f(a))
  \]
  is an embedding.
\end{definition}

Note that every equivalence is a (dependent) epimorphism.
Obviously, every dependent epimorphism is an epimorphism, but the converse holds
as well. In fact, this will be a consequence of our characterization of the
epimorphisms as the acyclic maps, to which we now turn.

\subsection{Acyclic maps}
The notion of an acyclic map is defined fiberwise using the suspension of a type
which we recall from~\cite[Sec.~6.5]{HoTTBook}:

\begin{definition}[\flinkspec{synthetic-homotopy-theory}{suspensions-of-types}{definitions} %
  Suspension \(\Sigma A\)]
  The \emph{suspension} of a type \(A\) is the pushout of the terminal maps
  \(\One \leftarrow A \rightarrow \One\) and we denote it by \(\Sigma A\).
  Equivalently, it is the higher inductive type generated by two points
  \(\north,\south : \Sigma A\) (north and south) and for every \(a : A\), an
  identification \(\merid_a : \north = \south\) (the meridians).
\end{definition}

\begin{definition}[\flinkspec{synthetic-homotopy-theory}{acyclic-types}{definition} %
  Acyclicity]\label{def:acyclic}
  A type is \emph{acyclic} if its suspension is contractible; a map is
  acyclic if all of its fibers are acyclic types.
  Note that a type \(A\) is acyclic precisely when the unique map \(A \to \One\)
  is acyclic.
\end{definition}

~\cref{sec:homology} explains the relation to the traditional formulation of
acyclicity (using reduced homology), while \cref{sec:examples-acyclic} presents
examples of acyclic types.
It is natural to consider variations on the notion of acyclicity where one instead
requires the \(n\)-fold suspension (for \(n \geq 2\)) to become
contractible. \cref{acyclic-stabilization} shows that these notions reduce to the
above notion of an acyclic type, at least in the presence of an additional
principle (\cref{sec:plus}).
For now, we work towards characterizing the
epimorphisms as the acyclic maps.

\begin{definition}[\flinkspec{synthetic-homotopy-theory}{codiagonals-of-maps}{definitions} %
  Codiagonal \(\nabla_f\)]\label{def:codiagonal}
  The \emph{codiagonal} \(\nabla_f\) of a map \(f : A \to B\) is the dashed map
  in the pushout diagram:
  \[
    \begin{tikzcd}
      A \ar[r,"f"] \ar[d,"f"'] \pocorner{}
      & B \ar[d,"\inr"] \ar[ddr,bend left=20,"\id"]  \\
      B \ar[drr,bend right=20,"\id"'] \ar[r,"\inl"']
      & B +_{A} B \ar[dr,dashed,"\nabla_f"] \\
      & & B
    \end{tikzcd}
  \]
\end{definition}

The reason for introducing the codiagonal is that it is the ``fiberwise
suspension'' as made precise by the following:
\begin{lemma}[\flinkspec{synthetic-homotopy-theory}{codiagonals-of-maps}{the-codiagonal-is-the-fiberwise-suspension}]%
  \label{codiagonal-is-fiberwise-suspension}
  For every \(f : A \to B\) and \(b : B\) we have an equivalence
  \(\fib_{\nabla_f}(b) \simeq \susp{\fib_f(b)}\) between the fiber of\/
  \(\nabla_f\) at \(b\) and the suspension of the fiber of\/~\(f\) at \(b\).
\end{lemma}
\begin{proof}
  By the flattening
  lemma~\cite[\href{https://unimath.github.io/agda-unimath/synthetic-homotopy-theory.flattening-lemma-pushouts.html}%
  {The flattening lemma for pushouts}]{agda-unimath}
  (cf.~\cite[Lem.~6.12.2]{HoTTBook}), we can pull back the pushout square
  in~\cref{def:codiagonal} along a point inclusion \(b : \One \to B\) to obtain
  the pushout square
  \[
    \begin{tikzcd}[row sep=5.5mm]
      \fib_f(b)
      \ar[d]
      \ar[r]
      \pocorner{}
      & \fib_{\id}(b) \ar[d] \\
      \fib_{\id}(b) \ar[r]
      & \fib_{\nabla_f}(b)
    \end{tikzcd}
  \]
  But the spans \(\fib_{\id}(b) \leftarrow \fib_f(b) \rightarrow \fib_{\id}(b)\)
  and \(\One \leftarrow \fib_f(b) \rightarrow \One\) are equivalent, so that
  \(\fib_{\nabla_f}(b)\) is also the pushout of the second span, i.e.\ it is the
  suspension of \(\fib_f(b)\), as desired.
\end{proof}

\begin{theorem}[\flinkspec{synthetic-homotopy-theory}{acyclic-maps}{a-map-is-acyclic-if-and-only-if-it-is-an-dependent-epimorphism} %
  Characterization of epimorphisms]\label{acyclic-characterization}
  The following are equivalent for a map \(f : A \to B\):
  \begin{enumerate}[label=(\roman*)]
  \item\label{item-epi} \(f\) is an epi,
  \item\label{item-dep-epi} \(f\) is a dependent epi,
  \item\label{item-acyclic} \(f\) is acyclic,
  \item\label{item-codiag-equiv} its codiagonal \(\nabla_f\) is an equivalence.
  \end{enumerate}
\end{theorem}
\begin{proof}
  The equivalence of \eqref{item-acyclic} and \eqref{item-codiag-equiv} follows
  from~\cref{codiagonal-is-fiberwise-suspension}.
  Moreover, \eqref{item-epi} and \eqref{item-codiag-equiv} are seen to be
  equivalent by~\cref{epic-iff-pushout-square}.
  Finally, suppose that \(f\) is an epi; we show that it a dependent epi as
  well. We need to prove that precomposition by \(f\) is an embedding
  \(\prod_{b : B}P(b) \hookrightarrow \prod_{a : A}P(f(a))\) for an arbitrary
  type family \(P\) over \(B\).
  The equivalence \(A \simeq \sum_{b : B}\fib_f(b)\) induces a commutative
  square
  \[
    \begin{tikzcd}
      \prod_{b : B}P(b) \ar[r,"f^\ast"] \ar[d]
      & \prod_{a : A}P(f(a)) \\
      \prod_{b : B}\pa*{\fib_f(b) \to P(b)} \ar[r,"\simeq"]
      & \prod_{w : \sum_{b : B}\fib_f(b)}P(\pi_1(w)) \ar[u,"\simeq"]
    \end{tikzcd}
  \]
  so it suffices to prove that the leftmost map, which is the functorial action
  of \(\prod\) at the constants map \(P(b) \to (\fib_f(b) \to P(b))\), is an
  embedding.
  By a special case of \cite[Exer.~13.12~(a)]{Rijke2022}, which is formalized
  in
  \cite[\href{https://unimath.github.io/agda-unimath/foundation.functoriality-dependent-function-types.html}%
  {Functoriality of dependent function types}]{agda-unimath},
  it suffices for this constants map to be an embedding for all \(b : B\).
  But this is indeed the case, because \(\fib_f(b)\) is acyclic, so that
  precomposition by the terminal map \(\fib_f(b) \to \One\) is an embedding.
\end{proof}

We will state further results in terms of acyclicity, often tacitly using that
the acyclic maps are exactly the epis.

For acyclic \emph{types}, we arrive at the following:

\begin{corollary}[\flinkspec{synthetic-homotopy-theory}{acyclic-maps}{a-type-is-acyclic-if-and-only-if-the-constant-map-from-any-type-is-an-embedding}]%
  \label{characterization-of-acyclic-types}
  The following are equivalent:
  \begin{enumerate}[label=(\roman*)]
  \item\label{item-type-acyclic} the type \(A\) is acyclic,
  \item\label{item-const-emb} for all types \(B\), the constants map
    \(B \to (A \to B)\) is an embedding,
  \item\label{item-const-equiv} for all types \(B\) and elements \(x,y : B\),
    the constants map \(x = y \to (A \to x = y)\) is an equivalence.
  \end{enumerate}
\end{corollary}
\begin{proof}
  Writing \(!_A : A \to \One\) for the unique map from \(A\) to the unit type,
  the commutative diagram
  \[
    \begin{tikzcd}
      (\One \to B) \ar[r,"(!_A)^\ast"] & (A \to B) \\
      B \ar[ur, "\const_{B,A}"'] \ar[u,"\simeq"]
    \end{tikzcd}
  \]
  informs us that \((!_A)^\ast\) is an embedding if and only if \(\const_{B,A}\)
  is, which proves the equivalence of \eqref{item-type-acyclic} and
  \eqref{item-const-emb} via \cref{acyclic-characterization}.
  For the equivalence of \eqref{item-const-emb} and \eqref{item-const-equiv} we
  let \(x,y : B\) and use the commutative diagram
  \[
    \begin{tikzcd}[row sep=4mm]
      (x = y) \ar[ddr,bend right,"\const_{x=y,A}"'] \ar[r,"\ap^{\const_{B,A}}_{x,y}"]
      &(\const_{B,A}(x) = \const_{B,A}(y))
      \ar[d,"\simeq"] \\
      &\prod_{a : A} \const_{B,A}(x)(a) = \const_{B,A}(y)(a)
      \ar[d,"\simeq"] \\
      &(A \to x = y)
    \end{tikzcd}
  \]
  to see the map \(\ap^{\const_{B,A}}_{x,y}\) is an equivalence if and only if
  \(\const_{x=y,A}\) is.
  But the former condition for all \(x,y : B\) says precisely that
  \(\const_{B,A}\) is an embedding.
\end{proof}

Examples of acyclic types, other than the unit type, are presented
in~\cref{sec:examples-acyclic}.

\subsection{Closure properties}\label{sec:acyclic-closure-properties}

We show that acyclic types and maps enjoy several closure
properties. By~\cref{acyclic-characterization} these properties also hold for
epimorphisms of course.

\begin{lemma}[\flinkspec{synthetic-homotopy-theory}{acyclic-maps}{the-class-of-acyclic-maps-is-closed-under-composition-and-has-the-right-cancellation-property}]%
  \label{acyclic-composition-and-right-cancellation}
  The class of acyclic maps is closed under composition and has the \emph{right
    cancellation property}: if \(g \circ f\) and \(f\) are acyclic, then so is
  \(g\).
\end{lemma}
\begin{proof}
  Suppose that \(f\) is acyclic, i.e., that precomposition with \(f\) is an
  embedding.
  Then, for any type \(X\), the composite
  \[
    (C \to X) \xrightarrow{g^\ast} (B \to X) \xrightarrow{f^\ast} (A \to X),
  \]
  which is \(\pa{g \circ f}^\ast\), is an embedding if and only if \(g^\ast\)
  is, by the composition and \emph{left} cancellation properties of the class of
  embeddings.
\end{proof}

It follows from~\cref{acyclic-composition-and-right-cancellation} that the acyclic types are closed under
dependent sums:

\begin{proposition}[\flinkspec{synthetic-homotopy-theory}{acyclic-maps}{acyclic-types-are-closed-under-dependent-pair-types}]%
  \label{acyclic-Sigma}
  If \(A\) is an acyclic type and \(B\) is a family of acyclic types over \(A\),
  then \(\sum_{a : A}B(a)\) is acyclic.
\end{proposition}
\begin{proof}
  The composite \(\sum_{a : A}B(a) \xrightarrow{\pi_1} A \to \One\) is acyclic,
  because \(A\) and \(\fib_{\pi_1}(a) \simeq B(a)\) are acyclic types.
\end{proof}

\begin{corollary}[\flinkspec{synthetic-homotopy-theory}{acyclic-maps}{acyclic-types-are-closed-under-binary-products}]
  Finite products of acyclic types are acyclic.
\end{corollary}
\begin{proof}
  This is the non-dependent version of the previous result.
\end{proof}

Another consequence is that (inhabited) locally acyclic types are acyclic:

\begin{lemma}[\flinkspec{synthetic-homotopy-theory}{acyclic-maps}{inhabited-locally-acyclic-types-are-acyclic}]%
  \label{acyclic-Id}
  If all identity types of an inhabited type~\(A\) are acyclic, then so is
  \(A\).
\end{lemma}
\begin{proof}
  Since being acyclic is a property, we may take a point \(a : A\).
  The composite \(\One \xrightarrow{\ulcorner{a}\urcorner} A \to \One\) is
  acyclic and so is \(\fib_{\ulcorner{a}\urcorner}(x) \simeq (a = x)\) for every
  \(x : A\) by assumption.
  But then \(A \to \One\) is acyclic too by~\cref{acyclic-composition-and-right-cancellation}.
\end{proof}

The converse of~\cref{acyclic-Id} fails, as we will see later when we have
produced an example of an acyclic type.

\begin{proposition}[\flinkspec{synthetic-homotopy-theory}{acyclic-maps}{acyclic-maps-are-closed-under-pushouts}]%
  \label{acyclic-pullback-and-pushout-stable}
  The acyclic maps are stable under pullbacks and pushouts.
\end{proposition}
\begin{proof}
  Given a pullback square
  \[
    \begin{tikzcd}
      A\times_C B \ar[very near start,"\lrcorner",phantom,dr]
      \ar[d,"\pi_1"'] \ar[r,"\pi_2"] & B \ar[d,"f"] \\
      A \ar[r,"g"'] & C
    \end{tikzcd}
  \]
  we have, for every \(a : A\), an equivalence 
  \(\fib_{\pi_1}(a) \simeq \fib_f(g(a))\) by~\cite[Lem.~7.6.8]{HoTTBook}, which
  proves that \(\pi_1\) is acyclic if \(f\) is.

  Assume $f$ is acyclic and consider the pushout diagram
  \[
    \begin{tikzcd}
      A \ar[d,"f"'] \ar[r,"g"] \pocorner{}
      & C \ar[d,"\inr"] \\
      B \ar[r,"\inl"']
      & B +_A C
    \end{tikzcd}
  \]
  The type of extensions of a map \(h \colon C \to X\) along \(\inr\) is
  equivalent to the type of extensions of \(h \circ g\) along \(f\) by the
  universal property of pushout.
  By acyclicity of \(f\), the latter extension problem has at most one solution
  (recall \cref{extensions-as-fiber}). But then so does the former, so \(\inr\)
  is acyclic.
\end{proof}

\begin{remark}
  Note that the acyclic types are not closed under coproducts: while \(\One\) is
  acyclic, the coproduct \(\One + \One\) is not, since
  \(\susp (\One + \One) \simeq \Circle\) which is, of course, not contractible.
\end{remark}

\begin{lemma}[\flinkspec{synthetic-homotopy-theory}{acyclic-types}{acyclic-types-are-closed-under-retracts}]%
  \label{acyclic-retract}
  Acyclic types are closed under retracts.
\end{lemma}
\begin{proof}
  If \(A\) is a retract of \(B\), then \(\Sigma A\) is a retract of \(\Sigma B\)
  by functoriality of the suspension. If \(B\) is acyclic, then
  \(\Sigma B \simeq \One\) and retracts of contractible types are contractible.
\end{proof}

\begin{corollary}
  The acyclic maps are also closed under retracts (in the sense
  of~\cite[Def.~4.7.2]{HoTTBook}).
\end{corollary}
\begin{proof}
  Since fibers of a retract are retracts of fibers~\cite[Lem.~4.7.3]{HoTTBook},
  this follows from \cref{acyclic-retract}.
\end{proof}

\begin{remark}
  The truncation of an acyclic type need not be acyclic.
  For a counterexample, we turn to classical K-theory and we don't spell out the
  details.
  The Volodin space $X(R)$ of a ring $R$ is the
  fiber at the base point of the acyclic map
  $\mathrm{BGL}(R) \to \mathrm{BGL}(R)^+$, and hence acyclic~\cite[Ex.~IV.1.3.2]{Kbook}.
  Its fundamental
  group \(\mathrm{St}(R)\) is the Steinberg group of $R$, and its higher
  homotopy groups contain the higher K-theory of $R$ (by the long exact sequence
  of a fibration). The Steinberg group is not acyclic, since
  $\mathrm{H}_3(\mathrm{St}(R)) \cong \mathrm{K}_3(R)$ \cite[Ex.~IV.1.9]{Kbook}, and,
  e.g., $\mathrm{K}_3(\mathbb{F}_q) \cong \mathbb{Z}/(q^2-1)$ by Quillen's computation
  of the K-theory of a finite field~\cite[Cor.~IV.1.13]{Kbook}.
  Since the $1$\nobreakdash-truncation of $X(R)$
  is the classifying space of the Steinberg group of $R$~\cite[Cor.~IV.1.7.2]{Kbook},
  this shows that the truncation of an acyclic type need not be acyclic.
\end{remark}

\subsection{Balanced maps}
We connect the acyclic maps to the notion of \emph{balanced
  maps} due to Raptis~\cite{Raptis2019}.
To do so, we first recall the construction of joins and smash
products~\cite[p.~33]{Brunerie2016}, relate them in
\cref{join-is-suspension-of-smash} and prove a general lemma about joins with an
acyclic type (\cref{join-with-acyclic-type-is-contractible}).

\begin{definition}[\flinkspec{synthetic-homotopy-theory}{joins-of-types}{definitions} %
  Join \(A \join B\)]\label{def:join}
  The \emph{join} of two types \(A\) and \(B\) is the pushout of the projections
  \(A \leftarrow A \times B \rightarrow B\) and we denote it by \(A \join B\).
\end{definition}

\begin{definition}[\flinkspec{synthetic-homotopy-theory}{smash-products-of-pointed-types}{definition} %
  Smash product \(A \smashpr B\)]\label{def:smash}
  The \emph{smash product} of two pointed types \(A\) and \(B\) is the pushout
  of the span \(\One \leftarrow A \wedgesum B \rightarrow A \times B\), where
  right leg is the induced (dashed) map in the following diagram:
  \[
    \begin{tikzcd}
      \One \pocorner \ar[r,"\pt_B"] \ar[d,"\pt_A"']
      & B \ar[ddr,bend left, "\pt_A \times \id_B"] \ar[d] \\
      A \ar[r] \ar[drr, bend right, "\id_A \times \pt_A"']
      & A \wedgesum B \ar[dr,dashed] \\
      & & A \times B
    \end{tikzcd}
  \]
  We denote the smash product by \(A \smashpr B\).
\end{definition}

\begin{lemma}\label{join-is-suspension-of-smash}
  For pointed types \(A\) and \(B\), their join is equivalent to the suspension
  of their smash product, i.e., \({A \join B} \simeq {\susp(A \smashpr B)}\).
\end{lemma}
\begin{proof}
  For an arbitrary pointed type \(X\), we have a sequence of natural
  equivalences
  \begin{align*}
    &\phantom{\hspace{14pt}} (\susp(A \smashpr B) \ptdto X) \\
    &\simeq
      (A \smashpr B \ptdto \Omega X)
    &\text{(loop-suspension adjunction \cite[Lem.~6.5.4]{HoTTBook})} \\
    &\simeq
      (A \ptdto (B \ptdto \Omega X))
    &\text{(by \cite[Thm.~4.3.28]{vanDoorn2018})} \\
    &\simeq
      (A \join B \ptdto X)
    &\text{(by \cite[Lem.~6.1]{CagneBuchholtzKrausBezem2024})}
  \end{align*}
  so \(A \join B\) and \(\susp(A \smashpr B)\) must be equivalent.
\end{proof}

\begin{lemma}\label{join-with-acyclic-type-is-contractible}
  The join of an acyclic type with an inhabited type is contractible.
\end{lemma}
\begin{proof}
  Suppose that \(A\) is acyclic and \(B\) is inhabited. In particular, \(A\) is
  also inhabited. Since we are proving a proposition, we may assume that both
  types are actually pointed so that \(A \join B \simeq \susp(A \smashpr B)\)
  by \cref{join-is-suspension-of-smash}.
  Now we just calculate:
  \begin{align*}
    \susp(A \smashpr B) &\simeq \Circle \smashpr (A \smashpr B)
    &\text{(by \cite[Prop.~4.2.1]{Brunerie2016})} \\
    &\simeq (\Circle \smashpr A) \smashpr B
    &\text{(by \cite[Prop.~9]{Ljungstrom2024})} \\
    &\simeq \susp A \smashpr B
    &\text{(by \cite[Prop.~4.2.1]{Brunerie2016})} \\
    &\simeq \One \smashpr B
    &\text{(by acyclicity of \(A\))} \\
    &\simeq \One,
    &\text{(straightforward)}
  \end{align*}
  so that \(A \join B\) is contractible, as desired.
\end{proof}

\begin{definition}[Balanced map]\label{def:balanced}
  A map $f : A \to B$ is \emph{balanced} if for every surjection $g : X \to B$,
  the pullback square
  \[
    \begin{tikzcd}
      F \ar[d]\ar[r]\pbcorner & A\ar[d,"f"] \\
      X \ar[r,"g"'] & B
    \end{tikzcd}
  \]
  is also a pushout square.
\end{definition}

\begin{lemma}[Pushouts are fiberwise]\label{pushouts-are-fiberwise}
  A square
  \[
    \begin{tikzcd}
      F \ar[d,"f'"']\ar[r,"g'"] & A\ar[d,"f"] \\
      X \ar[r,"g"'] & B
    \end{tikzcd}
  \]
  is a pushout square if and only if the induced squares of
  fibers
  \[
    \begin{tikzcd}
      \fib_{g \circ f'}(b) \ar[d]\ar[r] & \fib_f(b) \ar[d] \\
      \fib_g(b) \ar[r] & \One
    \end{tikzcd}
  \]
  are pushouts for every \(b : B\).
\end{lemma}
\begin{proof}
  If the first square is a pushout square, then by the flattening
  lemma~\cite[\href{https://unimath.github.io/agda-unimath/synthetic-homotopy-theory.flattening-lemma-pushouts.html}%
  {The flattening lemma for pushouts}]{agda-unimath}
  (cf.~\cite[Lem.~6.12.2]{HoTTBook}), we can pull it back along any point
  \(b : {\One \to B}\) to get that the induced square of fibers is a pushout.
  Conversely, suppose that each fiber square is a pushout and consider the cogap
  map~\(h\) in the diagram
  \[
    \begin{tikzcd}
      F \ar[d,"f'"']\ar[r,"g'"]\pocorner & A\ar[ddr,bend left,"f"] \ar[d] \\
      X \ar[drr,bend right,"g"'] \ar[r] & P \ar[dr,"h",dashed] \\
      & & B
    \end{tikzcd}
  \]
  We once again use the flattening lemma to obtain pushout squares
  \[
    \begin{tikzcd}
      \fib_{g \circ f'}(b) \ar[d]\ar[r]\pocorner & \fib_f(b) \ar[d] \\
      \fib_g(b) \ar[r] & \fib_h(b)
    \end{tikzcd}
  \]
  for every \(b : B\). By assumption, \(\fib_h(b) \simeq \One\), so \(h\) is an
  equivalence, as desired.
\end{proof}

The following result for the \(\infty\)-category of spaces is due to
Raptis~\cite[Thm.~2.1]{Raptis2019} where it is established using different
techniques.

\begin{theorem}\label{acyclic-iff-balanced}
  A map is acyclic if and only if it is balanced.
\end{theorem}
\begin{proof}
  In the forward direction, we note that \cref{pushouts-are-fiberwise} implies
  that it suffices to prove the proposition for acyclic maps into the unit type,
  i.e.\ when \(B \equiv \One\) in \cref{def:balanced}.
  In this case, \(F \simeq A \times X\) and \(P \simeq A \join X\)
  (recall~\cref{def:join}) and we have to show that the latter is contractible.
  But this follows from \cref{join-with-acyclic-type-is-contractible} because
  \(A\) is acyclic and \(X\) is inhabited (as \(g : X \to \One\) is assumed to
  be surjective).

  Conversely, take $X = \Two\times B$ and let $g : \Two \times B \to B$ be the
  projection.  Let $b:B$ be arbitrary and let $F$ be the fiber of $f$ at $b$.
  By pullback-stability of pushout squares, the pushout square on the left below
  pulls back to the pushout square on the right:
  \[
    \begin{tikzcd}
      \Two \times A\ar[r]\ar[d,"\Two \times f"']\pbcorner\pocorner & A\ar[d,"f"] \\
      \Two \times B\ar[r] & B
    \end{tikzcd}
    \qquad
    \begin{tikzcd}
      \Two \times F\ar[r]\ar[d]\pbcorner\pocorner & F\ar[d] \\
      \Two \ar[r] & \One
    \end{tikzcd}
  \]
  The square on the right being a pushout means
  $\One \simeq \Two \join F \simeq \Sigma F$, so $F$ is acyclic.
\end{proof}

The following appears for CW-spaces as
\cite[Thm~2.5]{HausmannHusemoller1979} and \cite[Cor.~2.10(b)]{Alonso1983}, and
is discussed in the context of the \(\infty\)-category of spaces on
\cite[p.~774]{Raptis2019}.

\begin{corollary}\label{fiber-is-cofiber}
  If $f : A \to B$ is an acyclic map of connected types, then the fiber sequence
  for any $b:B$ is also a cofiber sequence.
  That is, the pullback square
  \[
    \begin{tikzcd}[row sep=5mm,column sep=5mm]
      \fib_f(b) \ar[r] \ar[d] & A \ar[d,"f"] \\
      \One \ar[r,"b"] & B
    \end{tikzcd}
  \]
  is also a pushout square.
\end{corollary}

The dual Blakers--Massey theorem holds
for all left classes of modalities~\cite[Thm.~3.27]{ABFJ2020},
but the cited proof only uses stability under base change,
so it also holds for acyclic maps, as also observed by Raptis~\cite[Sec.~3.3]{Raptis2019},
so we record that as well.

\begin{proposition}[dual Blakers--Massey for acyclic maps]
  \label{dual-blakers-massey-acyclic}
  If in a pullback square
  \[
    \begin{tikzcd}
      Q \ar[d]\ar[r]\pbcorner & B\ar[d,"g"] \\
      A \ar[r,"f"'] & C
    \end{tikzcd}
  \]
  the pushout product $f \mathbin\square g$ is acyclic,
  then so is the cogap map $A +_Q B \to C$.
\end{proposition}

\section{Acyclic maps and epimorphisms of \texorpdfstring{\(k\)}{k}-types}\label{sec:k-acyclic}
It turns out that we can get nice characterizations and examples if we consider
a notion of epimorphism with respect to \(k\)-types only.

\begin{definition}[\flinkspec{foundation}{dependent-epimorphisms-with-respect-to-truncated-types}{definitions} %
  (Dependent) \(k\)-epimorphism]\label{def:k-epic}
  We define a map \({f : A \to B}\) to be a \emph{\(k\)\nobreakdash-epimorphism} if
  for all \(k\)-types \(X\), the precomposition map
  \(
    (B \to X) \xrightarrow{f^\ast} (A \to X)
  \)
  is an embedding.
  We say that \(f\) is a \emph{dependent \(k\)-epimorphism} if for every family
  \(P\) of \(k\)-types over~\(B\), the precomposition map
  \[
    \prod_{b : B}P(b) \xrightarrow{f^\ast} \prod_{a : A}P(f(a))
  \]
  is an embedding.
\end{definition}

Note that we did not require \(f\) to be a map between \(k\)-types in the above
definition. However, as the following result shows, being \(k\)-epic is stable
under \(k\)-truncation:

\begin{lemma}[\flinkspec{foundation}{epimorphisms-with-respect-to-truncated-types}{a-map-is-a-k-epimorphism-if-and-only-if-its-k-truncation-is-a-k-epimorphism}]%
  \label{k-epi-iff-epi-of-k-types}
  A map \(f : A \to B\) is \(k\)-epic if and only if its
  \(k\)-truncation \(\squash{f}_k : \squash{A}_k \to \squash{B}_k\) is.
\end{lemma}
\begin{proof}
  For every \(k\)-type \(X\), we have a commutative square
  \[
    \begin{tikzcd}
      (\squash*{B}_k\to X) \arrow[r,"\squash{f}_k^\ast"]
      \arrow[d,swap,"\tosquash{-}_k^\ast"]
      & (\squash*{A}_k\to X) \arrow[d,"\tosquash{-}_k^\ast"] \\
      (B\to X) \arrow[r,"f^\ast"] & (A\to X)
    \end{tikzcd}
  \]
  Since the vertical maps are equivalences by \cite[Lem.~7.5.7 and
  Cor.~7.5.8]{HoTTBook}, it follows that the top map is an embedding if and only
  if the bottom map is an embedding.
\end{proof}

We work towards characterizing the \(k\)-epimorphisms.

\begin{definition}[\flinkspec{synthetic-homotopy-theory}{truncated-acyclic-types}{definition} %
  \(k\)-acyclicity]
  A type is \emph{\(k\)-acyclic} if its suspension is \(k\)-connected (i.e., the
  \(k\)-truncation of its suspension is contractible).
  A map is \(k\)-acyclic if all of its fibers~are.
\end{definition}

Note that every \((k+1)\)-acyclic type is \(k\)-acyclic. Since suspensions are
inhabited, every type is \((-1)\)-acyclic.

As with acyclic types, it is natural to consider variations on the notion of
\(k\)\nobreakdash-acyclicity where one instead requires the \(n\)-fold
suspension (for \(n \geq 2\)) to become
\(k\)-connected. In \cref{k-acyclic-stabilization} we show that these notions
suitably reduce to the above notion.

We recall the notion of a \(k\)-equivalence from~\cite{CORS2020} (where the
notion was introduced for an arbitrary modality):

\begin{definition}[\flinkspec{foundation}{truncation-equivalences}{definition} \(k\)-equivalence]
  A map is a \emph{\(k\)-equivalence} if its \(k\)-truncation is an equivalence.
\end{definition}

In general, not every \(k\)-equivalence is \(k\)-connected. For example, the
``degree 2'' map \({d_2 : \Circle \to \Circle}\) that sends \(\loopc\) to
\(\loopc \bullet \loopc\) is a 0-equivalence, but is not 0-connected.
Alternatively, any map \(\One \to \Two\) is a \((-1)\)-equivalence but no such
map is \((-1)\)-connected (=~surjective).
However, for retractions the notions \emph{do} coincide and this observation
proves very useful in the characterization of \(k\)-epis.

\begin{lemma}[\flinkspec{foundation}{truncation-equivalences}{a-k-equivalence-with-a-section-is-k-connected}]%
  \label{retractions-equivalence-connected}
  A retraction is a \(k\)-equivalence if and only if it is \(k\)-connected.
\end{lemma}
\begin{proof}
  For \(k \geq 0\), this follows from~\cite[Cor.~8.8.5]{HoTTBook} together with
  the fact that if \(r \circ s = \id\), then
  \(\pi_n(r) \circ \pi_n(s) = \id\), so \(\pi_n(r)\) must be surjective.

  For \(k = -1\), it follows from the fact that retractions are
  \((-1)\)-connected (=~surjective), and for \(k = -2\), the claim is trivial.%
  \footnote{More generally, using~\cite[Prop.~2.31]{CORS2020} one can prove that
    a section is an \(L'\)-equivalence if and only if it is \(L'\)-connected for
    any localization \(L\).}
\end{proof}

\begin{theorem}[\flinkspec{synthetic-homotopy-theory}{truncated-acyclic-maps}{a-map-is-k-acyclic-if-and-only-if-it-is-an-dependent-k-epimorphism} %
  Characterization of \(k\)-epimorphisms]\label{characterization-of-k-epis}
  For a map \(f : A \to B\), the following are equivalent:
  \begin{enumerate}[label=(\roman*)]
  \item\label{item-k-epi} \(f\) is a \(k\)-epi,
  \item\label{item-codiagonal-k-equiv} its codiagonal\/ \(\nabla_f\) is a
    \(k\)-equivalence,
  \item\label{item-codiagonal-k-connected} its codiagonal\/ \(\nabla_f\) is
    \(k\)-connected,
  \item\label{item-k-acyclic} \(f\) is \(k\)-acyclic,
  \item\label{item-dep-k-epi} \(f\) is a dependent \(k\)-epi.
  \end{enumerate}
\end{theorem}
\begin{proof}
  \eqref{item-k-epi}~\(\iff\)~\eqref{item-codiagonal-k-equiv}:
  By replaying the proof of~\cref{epic-iff-pushout-square}, we see that \(f\)
  is a \(k\)-epi if and only if the square
  \[
    \begin{tikzcd}
      \squash{A}_k \ar[r,"\squash{f}_k"] \ar[d,"\squash{f}_k"']
      & \squash{B}_k \ar[d,"\id"] \\
      \squash{B}_k \ar[r,"\id"'] & \squash{B}_k
    \end{tikzcd}
  \]
  is a pushout of \(k\)-types (recall~\cite[Thm.~7.4.12]{HoTTBook}).
  But this just means that \(\squash{\nabla_f}_k\) is an equivalence, i.e., that
  \(\nabla_f\) is a \(k\)-equivalence.

  \eqref{item-codiagonal-k-equiv}~\(\Longrightarrow\)~\eqref{item-codiagonal-k-connected}:
  This follows at once from~\cref{retractions-equivalence-connected}, because
  the codiagonal is, by definition, a retraction.

  \eqref{item-codiagonal-k-connected}~\(\Longrightarrow\)~\eqref{item-k-acyclic}:
  This holds because the codiagonal is the fiberwise suspension
  (\cref{codiagonal-is-fiberwise-suspension}).

  \eqref{item-k-acyclic}~\(\Longrightarrow\)~\eqref{item-dep-k-epi}:
  Suppose that \(f\) is \(k\)-acyclic; we show that \(f\) is a dependent
  \(k\)-epi.
  Let \(P\) be a family of \(k\)-types over \(B\). %
  As in the proof of~\cref{acyclic-characterization}, it suffices to show that
  we have embeddings
  \(P(b) \hookrightarrow (\fib_f(b) \to P(b))\) for every \(b : B\).
  For every \(b : B\), the type \(\fib_f(b)\) is \(k\)-acyclic by assumption, so
  the codiagonal of \(\fib_f(b) \to \One\) is a \(k\)-equivalence and hence the
  map \(\fib_f(b) \to \One\) is a \(k\)-epi by the equivalence
  of~\eqref{item-k-epi} and \eqref{item-codiagonal-k-equiv} proved above.
  Thus, since every \(P(b)\) is a \(k\)-type, we have the desired embeddings
  \(P(b) \hookrightarrow (\fib_f(b) \to P(b))\).

  \eqref{item-dep-k-epi}~\(\Longrightarrow\)~\eqref{item-k-epi}: By specializing
  to non-dependent functions.
\end{proof}

As with~\cref{characterization-of-acyclic-types}, the characterization theorem implies:

\begin{corollary}[\flinkspec{synthetic-homotopy-theory}{truncated-acyclic-maps}{a-type-is-k-acyclic-if-and-only-if-the-constant-map-from-any-identity-type-of-any-k-type-is-an-equivalence}]%
  \label{characterization-of-k-acyclic-types}
  The following are equivalent:
  \begin{enumerate}[label=(\roman*)]
  \item the type \(A\) is \(k\)-acyclic,
  \item for all \(k\)-types \(B\), the constants map \(B \to (A \to B)\) is an
    embedding,
  \item for all \(k\)-types \(B\) and \(x,y : B\), the constants map
    \(x = y \to (A \to x = y)\) is an equivalence.
  \end{enumerate}
\end{corollary}

\begin{corollary}[\flinkspec{synthetic-homotopy-theory}{0-acyclic-maps}{a-map-is-0-acyclic-if-and-only-if-it-is-surjective}]
  A type/map is \(0\)-acyclic if and only if it is \((-1)\)-connected.
\end{corollary}
\begin{proof}
  The epimorphisms of sets are precisely the surjective maps, i.e., those maps
  with \((-1)\)-connected fibers, so this follows
  from~\cref{characterization-of-k-epis}.
\end{proof}

The following result implies that the \(k\)-sphere \(\Sphere{k}\) is an example
of a \((k+1)\)-acyclic type.

\begin{proposition}[\flinkspec{synthetic-homotopy-theory}{truncated-acyclic-maps}{every-k-connected-map-is-k1-acyclic}]%
  \label{k-equivalence-connected-acyclic}
  Every \(k\)-connected map is \((k+1)\)-acyclic and \(k\)\nobreakdash-equivalences are \(k\)-acyclic.
\end{proposition}
\begin{proof}
  The first claim follows because taking the suspension increases the
  connectedness by one~\cite[Thm.~8.2.1]{HoTTBook}.
  Notice that the other claim is trivial for \(k = -2\). For \(k = k' + 1\) with
  \(k' \geq -2\), we have that a \(k\)-equivalence is \(k'\)-connected
  by~\cite[Prop.~2.30]{CORS2020}, and hence \(k\)-acyclic by the above.
\end{proof}

These implications are strict as shown in~\cref{2-acyclic-but-not-1-connected}.
However, for simply-connected (= \(1\)-connected) maps/types we do have the
following result:

\begin{proposition}\label{k-acyclic-iff-connected}
  A simply connected type is \((k+1)\)-acyclic if and only if it is
  \(k\)-connected.
\end{proposition}
\begin{proof}
  The right-to-left implication is~\cref{k-equivalence-connected-acyclic}.
  For the converse, we may assume \(k\ge1\) and we use the Freudenthal suspension
  theorem~\cite[Thm.~8.6.4]{HoTTBook}: The unit map of the loop-suspension
  adjunction, \(\sigma : A \to \Omega\Sigma A\), is \(2n\)-connected whenever
  \(A\) is \(n\)-connected (for \(n\ge0\)).
  Since \(A\) is assumed to be \((k+1)\)-acyclic, the map
  \(\Omega\Sigma A \to \One\) is \(k\)-connected.
  Hence, starting with the assumption that \(A\) is \(1\)\nobreakdash-connected,
  we in turn conclude that \(A\) is \(\min(2,k)\)-connected, then
  \(\min(4,k)\)-connected, etc., and hence \(k\)-connected.
\end{proof}

A consequence of the above result is that joins of \(k\)-acyclic types become
\(2k\)\nobreakdash-connected, as we show after this lemma.

\begin{lemma}\label{suspension-of-join-is-smash-of-suspensions}
  The suspension of the join of two pointed types \(A\) and \(B\) is the smash
  products of their suspensions, i.e.,
  \(\susp(A \join B) \simeq {\susp A} \smashpr {\susp B}\).
\end{lemma}
\begin{proof}
  We have the chain of equivalences:
  \begin{align*}
    \susp(A \join B) &\simeq \susp^2(A \smashpr B)
    &\text{(by \cref{join-is-suspension-of-smash})} \\
    &\simeq \Circle \smashpr \Circle \smashpr (A \smashpr B)
    &\text{(by \cite[Prop.~4.2.1]{Brunerie2016})} \\
    &\simeq \Circle \smashpr (A \smashpr B) \smashpr \Circle
    &\text{(the smash product is clearly commutative)} \\
    &\simeq (\Circle \smashpr A) \smashpr (B \smashpr \Circle)
    &\text{(by \cite[Prop.~9]{Ljungstrom2024})} \\
    &\simeq \susp A \smashpr \susp B
    &\text{(by \cite[Prop.~4.2.1]{Brunerie2016})}. &\qedhere
  \end{align*}
\end{proof}

\begin{proposition}\label{join-acyclic}
  If \(A\) is a \(k\)-acyclic type and \(B\) is an \(l\)-acyclic type with
  \(k,l \geq 0\), then their join \(A \join B\) is \((k + l)\)-connected.
  In particular, the join of two \(k\)-acyclic types is \(2k\)-connected.
\end{proposition}
\begin{proof}
  Since \(k,l \geq 0\), the types \(A\) and \(B\) are inhabited, and since we
  are proving a proposition we may assume them to be pointed so that
  \cref{suspension-of-join-is-smash-of-suspensions} applies and
  \(\susp(A \join B) \simeq {{\susp A} \smashpr {\susp B}}\).
  Now, \(\susp A\) is \(k\)-connected and \(\susp B\) is \(l\)-connected by
  assumption so that their smash product is \((k + l + 1)\)-connected
  by~\cite[Prop.~4.3.1]{Brunerie2016}, indeed showing that \(A \join B\) is
  \((k + l)\)-connected.
\end{proof}

\subsection{Closure properties}
Before characterizing 1- and 2-acyclic types, we record a few general closure
properties of \(k\)-acyclic maps which are proved analogously to the ones for
acyclic maps.

\begin{lemma}[\flinkspec{synthetic-homotopy-theory}{truncated-acyclic-maps}{the-class-of-k-acyclic-maps-is-closed-under-composition-and-has-the-right-cancellation-property}]%
  \label{k-acyclic-composition-and-right-cancellation}
  The class of \(k\)-acyclic maps is closed under composition and has the
  right cancellation property: if \(g \circ f\) and \(f\) are
  \(k\)-acyclic, then so is \(g\).

  Moreover, the \(k\)-acyclic maps are closed under retracts, pullbacks and
  pushouts.

  Finally, \(k\)-acyclic types are closed under \(\sum\)-types, and if all
  identity types of an inhabited type \(A\) are \(k\)-acyclic, then so is \(A\)
  itself.
\end{lemma}

\begin{corollary}[\flinkspec{synthetic-homotopy-theory}{truncated-acyclic-maps}{a-type-is-k1-acyclic-if-and-only-if-its-k-truncation-is}]%
  \label{k+1-acyclic-iff-k-truncation}
  A type is \((k+1)\)-acyclic if and only if its \(k\)-truncation is.
\end{corollary}
\begin{proof}
  The first map in the composite
  \(A \xrightarrow{\tosquash{-}_k} \squash{A}_k \to \One\) is
  \(k\)-connected~\cite[Cor.~7.5.8]{HoTTBook} and so
  \((k+1)\)\nobreakdash-acyclic by~\cref{k-equivalence-connected-acyclic}.  The
  result now follows from \cref{k-acyclic-composition-and-right-cancellation}.
\end{proof}

\begin{remark}
  The \(k\)-acyclic types are not closed under taking exponentials: the circle
  \(\Circle\) is 1-acyclic (as it is connected), but
  \((\Circle \to \Circle) \simeq \Circle \times \mathbb Z\) is not by the
  upcoming~\cref{1-acyclic-characterization}.
\end{remark}

\subsection{Characterizing \texorpdfstring{\(1\)}{1}-acyclic types and applications in group theory}
The following theorem says that there are no interesting 1-acyclic
\emph{sets} and is used to characterize the 1-acyclic types.
\begin{theorem}[\flinkspec{synthetic-homotopy-theory}{1-acyclic-types}{every-1-acyclic-type-is-0-connected}]%
  \label{1-acyclic-set-iff-contractible}
  For a \emph{set} \(A\), we have
  \[
    A \text{ is 1-acyclic} \iff A \text{ is acyclic} \iff A \text{ is contractible}.
  \]
\end{theorem}
\begin{proof}
  The right-to-left implications are trivial, so it suffices to show that every
  1-acyclic set is contractible.
  If \(A\) is 1-acyclic, then by~\cref{characterization-of-k-acyclic-types}, the
  constants map
  \[
    \Omega(\B G,\basept) \to (A \to \Omega(\B G,\basept))
  \]
  must be an equivalence, where \(\B G\) is the classifying
  type~\cite{BDR2018,Symmetry} of the free group \(G\) on the set \(A\).
  Now, \(\Omega(\B G,\basept) \simeq G\), so in particular, the unit
  \(\eta : A \to G\) of the free group \(G\) must be constant. Hence,
  \(\eta(x) = \eta(y)\) for every \(x,y : A\). But the unit is
  left-cancellable~\cite[Ch.~X]{MinesRichmanRuitenburg1988}
  (see~\cite{EscardoFreeGroup} or \cite[Ex.~11]{Warn2023} for a proof in
  homotopy type theory), so \(A\) is a proposition. Finally, if \(A\) is
  1-acyclic, then it is 0-acyclic, i.e., inhabited. Thus, \(A\) is contractible,
  as we wished to show.
\end{proof}

\begin{theorem}[\flinkspec{synthetic-homotopy-theory}{1-acyclic-types}{every-1-acyclic-type-is-0-connected} %
  Characterization of 1-acyclic types]\label{1-acyclic-characterization}
  A type is 1-acyclic if and only if it is connected.
\end{theorem}
\begin{proof}
  All connected types are 1-acyclic by~\cref{k-equivalence-connected-acyclic}.
  Conversely, if \(A\) is 1-acyclic, then the composite
  \(A \xrightarrow{\tosquash{-}_0} \squash{A}_0 \to \One\) is 1-acyclic. But the
  first map is 1-acyclic by~\cref{k+1-acyclic-iff-k-truncation}, so that
  \(\squash{A}_0\) is 1-acyclic by the right-cancellation property of the class
  of \(k\)-acyclic maps.
  But \(\squash{A}_0\) is a set by definition, so \(A\) is connected
  by~\cref{1-acyclic-set-iff-contractible}.
\end{proof}

It follows directly from \cref{1-acyclic-characterization} that a map is
$1$-acyclic if and only if it is connected. Combined with
\cref{characterization-of-k-epis} this can be used to give a constructive proof
of the following fact about group homomorphisms.

\begin{theorem}\label{epi-of-groups-characterization}
  The following are equivalent for a group homomorphism \(f : G \to H\):
  \begin{enumerate}[label=(\roman*)]
  \item\label{item-epi-of-grps} \(f\) is an epi of groups;
  \item\label{item-epi-of-concrete-grps} \(\B f : \B G \ptdto \B H\) is an epi
    of pointed connected \(1\)-types;
  \item\label{item-connected} \(\B f : \B G \to \B H\) is connected;
  \item\label{item-surjective} \(f\) is surjective as a map of sets.
  \end{enumerate}
\end{theorem}
\begin{proof}
  The implication
  \(\eqref{item-epi-of-grps} \Rightarrow \eqref{item-epi-of-concrete-grps}\)
  follows immediately from the equivalence between the category of groups
  and the category of pointed connected 1-types~\cite[Thm.~5.1]{BDR2018}.

  To see that \eqref{item-epi-of-concrete-grps} implies \eqref{item-connected},
  we suppose that \(\B f\) is an epi in the category of pointed connected
  \(1\)-types and pointed maps, and show that is also an epi of \(1\)-types,
  which by \cref{characterization-of-k-epis,1-acyclic-characterization} is
  equivalent to being connected.
  Assume we are given a map \(g : \B G \to X\) whose codomain~\(X\) is a
  \(1\)-type. We are to show that \(g\) has at most one extension along
  \(\B f\) (recall \cref{extensions-as-fiber}).
  Now we can make \(g\) a pointed map by pointing \(X\) at
  \(\pt_X \colonequiv g(\pt_{\B G})\).
  By connectedness of \(\B G\), the map \(g\) factors through the connected
  component \(X_0 \colonequiv \sum_{x : X} \squash{x = \pt_X}\) of \(\pt_X\) as
  a map \(g_0 : \B G \to X_0\).
  Because \(\B H\) is connected and the maps \(\B f\) and \(g_0\) are pointed,
  the type of bare extensions of \(g\) along \(\B f\) is equivalent to
  \(\sum_{e : \B H \to X} (e \circ \B f \sim g) \times \prod_{z : \B
    H}\squash{e(z) = \pt_X}\), which is in turn equivalent to
  \(\sum_{e_0 : \B H \to X_0}(e_0 \circ \B f \sim g_0)\),
  i.e., the type of extensions of \(g_0\) along \(\B f\).
  Now the type of \emph{pointed} extensions of \(g_0\) along \(\B f\)
  is equivalent to the type
  \[
    \sum_{e_0 \,:\, \B H \to X_0}\,\sum_{r \,:\, {e_0(\pt_{\B H})} \,=\, {\pt_X}}\,
    \sum_{H \,:\, {e_0 \,\circ\, \B f} \,\sim\, g_0}
    (q \bullet H(\pt_{\B G}) \bullet \ap_{e_0}(p)^{-1} = r),
  \]
  which, by contracting \(r\) with its identification, is equivalent to the type
  of bare extensions of \(g_0\) along \(\B f\) as ordinary maps between
  (connected) types.
  Thus, if \(\B f : \B G \ptdto B H\) is epic in the category of pointed
  connected \(1\)-types, then \(\B f\) is 1-epic and hence connected (by
  \cref{characterization-of-k-epis,1-acyclic-characterization}).

  If \(\B f\) is connected, then \(f\) is surjective as a map of sets
  by~\cite[Lem.~4.11.4]{Symmetry}, so
  \(\eqref{item-connected} \Rightarrow \eqref{item-surjective}\).
  Finally, \eqref{item-surjective} straightforwardly implies \eqref{item-epi-of-grps}.
\end{proof}

\begin{remark}\label{flattening-avoids-excluded-middle}
  Many traditional proofs of the fact that the epimorphisms of groups are
  precisely the surjections rely on excluded middle. For instance, the suggested
  proof in Mac~Lane's~\cite[Exer.~I.5.5]{cwm} relies heavily on a case analysis that
  requires the law of excluded middle.
  A notable exception is Todd Trimble's proof~\cite{Trimble2020} which is
  constructive and uses group actions. The above gives a different constructive
  proof relying instead on deloopings of groups and flattening (via the
  characterization of \(k\)-epis in \cref{characterization-of-k-epis}).
\end{remark}



We now give an application of \cref{epi-of-groups-characterization} by
presenting several structural characterizations of a group being generated by a
subset.

\begin{corollary}
  Given an injection \(\iota : S \hookrightarrow G\) from a set \(S\) to a group
  \(G\), the following are equivalent:
  \begin{enumerate}[label=(\roman*)]
  \item\label{item-S-generates-G} \(S\) generates \(G\);
  \item\label{item-hom-eval} for every group \(H\), the map
    \begin{align*}
      \iota^\ast : \Grp(G,H) &\to (S \to H) \\
      \varphi &\mapsto (s \mapsto \varphi(\iota(s)))
    \end{align*}
    which evaluates a group homomorphism \(\varphi\) is an embedding;
  \item\label{item-induced-map-surjective} the map \(\hat\iota : \F S \to G\)
    from the free group generated by \(S\) to \(G\), induced by \(\iota\), is
    surjective;
  \item\label{item-induced-map-connected} the map
    \(\B{\hat\iota} : \BF S \to \B G\) is connected.
  \end{enumerate}
\end{corollary}
\begin{proof}
  For \eqref{item-S-generates-G} \(\Rightarrow\) \eqref{item-hom-eval}, note
  that if every element of \(G\) is a finite combination of elements in \(S\)
  and their inverses, then a group homomorphism from \(G\) to another group
  \(H\) is completely determined by its effect on \(S\).
  Items \eqref{item-induced-map-surjective} and
  \eqref{item-induced-map-connected} are equivalent by
  \cref{epi-of-groups-characterization}.
  Moreover, \eqref{item-S-generates-G} and \eqref{item-induced-map-surjective}
  is are clearly equivalent.

  Finally, to see that \eqref{item-hom-eval} and
  \eqref{item-induced-map-connected} are equivalent, we consider the commutative
  diagram
  \[
    \begin{tikzcd}
      \Grp(G,H) \ar[r,"\iota^\ast"] \ar[d,"\simeq"]
      & (S \to H) \\
      (\B G \ptdto \B H) \ar[r,"(\B{\hat\iota})^\ast"]
      & ({\BF S} \ptdto \B H) \ar[u,"\simeq"']
    \end{tikzcd}
  \]
  where the vertical maps are equivalences by \cite[Thm.~5.1]{BDR2018}.
  This diagram tells us that \(\iota^\ast\) is an embedding if and only if
  \((\B{\hat\iota})^\ast\) is.
  But the latter happens exactly when \(\B{\hat\iota}\) is an epi in the
  category of pointed connected 1-types and pointed maps, which, as we saw in
  \cref{epi-of-groups-characterization}, is equivalent to \(\B{\hat\iota}\)
  being connected.
\end{proof}



\subsection{Characterizing \texorpdfstring{\(2\)}{2}-acyclic
  types}\label{sec:characterizing-2-acyclic-types}
The notion of $2$-acyclicity turns out to be closely related to perfect groups.
Most textbooks, e.g.~\cite[Exer.~19,Sec.~5.4]{AA2004}, define a group \(G\) to be
perfect if it equals its own commutator subgroup \(G'\).
Since the abelianization~\cite[Prop.~7,Sec.~5.4]{AA2004} of a group \(G\) is given
by the quotient \(G/G'\), we can reformulate perfectness as follows:

\begin{definition}[Perfectness]
  A group is \emph{perfect} if its abelianization is trivial.
\end{definition}

An example of a perfect group is the alternating group \(A_5\) on 5 generators.
Given a group \(G\), we recall from~\cite[Sec.~6]{BDR2018} that 2-truncating
the suspension of the classifying type \(\B G\) of \(G\) gives the classifying
type of the abelianization of \(G\) as an abelian group, i.e.\
\(\squash{\susp\B G}_2\) is the second delooping of its set of elements.
This immediately yields the following result:

\begin{proposition}\label{2-acyclic-iff-perfect}
  The classifying type of a group \(G\) is \(2\)\nobreakdash-acyclic if
  and only if \(G\) is perfect. \hfill\qedsymbol
\end{proposition}

\begin{remark}
  The classifying type of a group is always $1$-acyclic
  by~\cref{1-acyclic-characterization}.
\end{remark}

\begin{remark}\label{2-acyclic-but-not-1-connected}
  While \cref{k-equivalence-connected-acyclic} tells us that every
  \(k\)-equivalence is \(k\)-acyclic, the converse fails. In fact, even a
  \((k+1)\)-acyclic map need not be a \(k\)-equivalence as \(\B G \to \One\) is
  a 1-equivalence if and only if \(G\) is trivial, but it is a $2$-acyclic map if
  and only if \(G\) is perfect by~\cref{2-acyclic-iff-perfect}.
\end{remark}

\begin{theorem}[Characterization of $2$-acyclic types]\label{2-acyclic-type-char}
  A type \(A\) is $2$-acyclic if and only if \(A\) is connected and \(\pi_1(A,a)\)
  is perfect for every \(a : A\).
\end{theorem}
\begin{proof}
  Note that connectedness is necessary, because 1-acyclic types are connected.
  Moreover, if \(A\) is connected, then \(\B \pi_1(A,a) = \squash{A}_1\) for
  every \(a : A\). By~\cref{k+1-acyclic-iff-k-truncation}, the type \(A\) is
  $2$-acyclic if and only if \(\squash{A}_1\) is. So
  by~\cref{2-acyclic-iff-perfect}, this happens exactly when \(\pi_1(A,a)\) is
  perfect.
\end{proof}

Connected maps preserve $2$-acyclicity:
\begin{corollary}
  If \(f : A \to B\) is connected and \(A\) is $2$\nobreakdash-acyclic, then so is
  \(B\) (which is equivalent to the image of~\(f\) by connectedness).
\end{corollary}
\begin{proof}
  Note that \(B\) is connected, because \(A\) and \(f\) are. So
  by~\cref{2-acyclic-type-char} it suffices to prove that \(\pi_1(B,b)\) is
  perfect for every \(b : B\). Given \(b : B\), there exists \(a : A\) with
  \(f(a) = b\) as \(f\) is connected.
  By connectedness of \(f\) and~\cite[Cor.~8.4.8(ii)]{HoTTBook}, the map
  \(\pi_1(f,a) : \pi_1(A,a) \to \pi_1(B,b)\) is a surjection for every
  \(a : A\). But \(\pi_1(A,a)\) is perfect and any quotient of a perfect group
  is perfect, so \(\pi_1(B,b)\) is also perfect.
\end{proof}

The following gives a necessary condition on the fundamental group functor for
$2$-acyclicity:
\begin{proposition}\label{necessary-condition-for-2-acyclic-map}
  If \(f : A \to B\) is $2$-acyclic, then it is connected, and for every
  \(a : A\), we have that \(\ker(\pi_1(f,a))\) is perfect and the abelianization
  of \(\pi_1(f,a)\) is an isomorphism.
\end{proposition}
\begin{proof}
  If \(f : A \to B\) is $2$-acyclic, then it is certainly connected by the
  characterization of $2$-acyclic types. Let \(a : A\) be arbitrary and write
  \(F \colonequiv \fib_f(f(a))\). We have an exact
  sequence~\cite[Sec.~8.4]{HoTTBook}
  \[
    \pi_1(F) \xrightarrow{\pi_1(i)} \pi_1(A) \xrightarrow{\pi_1(f,a)} \pi_1(B)
    \to \pi_0(F) \cong \One,
  \]
  where the equivalence at the end holds because \(f\) is connected.
  Now \(\ker(\pi_1(f,a)) = \im(\pi_1(i))\) and \(\pi_1(F)\) is perfect, because
  \(F\) is $2$-acyclic. But any quotient of a perfect group is perfect, and hence,
  \(\im(\pi_1(i)) = \ker(\pi_1(f))\) is perfect.
  Finally, because abelianization is right exact (being a left adjoint,
  abelianization preserves all colimits), the exact sequence induces another
  exact sequence
  \[
    \pi_1(F)^{\text{ab}} \xrightarrow{\pi_1(i)^{\text{ab}}} \pi_1(A)^{\text{ab}}
    \xrightarrow{\pi_1(f,a)^{\text{ab}}} \pi_1(B)^{\text{ab}} \to \One.
  \]
  But \(\pi_1(F)^{\text{ab}}\) is trivial since \(\pi_1(F)\) is perfect, so
  the middle map \(\pi_1(f,a)^\text{ab}\) is an isomorphism.
\end{proof}

We remark that the conditions in \cref{necessary-condition-for-2-acyclic-map} are
not \emph{sufficient} for deriving 2-acyclicity: for example, the base point map
\(\One \to \Sphere{2}\) satisfies the conditions of the proposition but is not
\(2\)-acyclic, as its fiber at the base point is \(\Omega\Sphere{2}\) whose
fundamental group \(\mathbb Z\) is not perfect.
At present, we do not know of a necessary and sufficient characterization.

\subsection{Iterated suspensions and stabilization}
As mentioned before, it is natural to consider variations on the notion of
\(k\)\nobreakdash-acyclicity where one instead requires the \(n\)-fold
suspension (for \(n \geq 2\)) to become \(k\)-connected.
We show that these notions suitably reduce to the notion of
\(k\)-acyclicity. More precisely, we have a stabilization result which says that
the \(n\)-fold suspension of a type \(X\) is \(k\)-connected if and only if
\(X\) is \((k - n + 1)\)-acyclic, for \(n \geq 1\) and \(k \geq 2\).

\begin{lemma}\label{suspension-2-acyclic-iff-type-connected}
  The suspension of a type \(X\) is 2-acyclic if and only if \(X\) is connected.
\end{lemma}
\begin{proof}
  If \(X\) is connected, then \(\susp X\) is 2-acyclic since suspensions increase
  connectedness by one~\cite[Thm~8.2.1]{HoTTBook}.
  For the converse, suppose that \(\susp X\) is 2-acyclic. In particular,
  \(\pi_2(\susp^2 X)\) is trivial. The unit of the set truncation
  \(X \to \squash{X}_0\) induces a map \(\susp^2 X \to \susp^2 {\squash{X}_0}\)
  whose fibers are \(1\)-connected because its domain is \(1\)-connected (as
  \(X\) is inhabited) and its codomain is \(2\)-connected. The long exact
  sequence~\cite[Sec.~8.4]{HoTTBook} of this map then tells us that
  \(\pi_2(\susp^2 {\squash{X}_0})\) is also trivial.
  By~\cite[Sec.~6]{BDR2018}, the abelian group \(\pi_2(\susp^2 {\squash{X}_0})\) is
  in fact the free abelian group on the \emph{pointed} set \(\squash{X}_0\).
  Below we describe an adaptation (to pointed sets) of an argument due to David
  W\"arn~\cite{UnitFreeAbelianGroup} to prove that the unit of the
  free-forgetful adjunction between abelian groups and pointed sets is
  injective. This implies that \(\squash{X}_0\) injects into the trivial group
  \(\pi_2(\susp^2 {\squash{X}_0})\), showing that \(X\) is indeed connected.

  The central idea, going back to Roswitha Harting~\cite{Harting1982} and
  recently also used in homotopy type theory by Jarl Taxer\r{a}s Flaten
  in~\cite{TaxerasFlaten2022}, is to regard a (pointed) set as a filtered
  colimit.
  For a pointed set \(\pt_X : X\) we consider the following category \(I_X\):
  Its objects are pairs \((n,f)\) with \(n\) a natural number and
  \(f : [n] \ptdto X\), where \([n]\) is the standard \((n + 1)\)-element
  set pointed at \(0\).
  A morphism between such objects \((n,f)\) and \((m,g)\) is a pointed map
  \(p : [n] \ptdto [m]\) such that \(f = g \circ p\).
  One can show that \(I_X\) is a filtered category and that
  \(X \cong \colim_{(n,f) : I_X}[n]\).

  The functor \(F\) that produces the free abelian group on a pointed set
  preserves colimits, as it is a left adjoint, so
  \(FX \cong \colim_{(n,f) : I_X}[n]\).
  The forgetful functor sending an abelian group to its underlying set preserves
  filtered colimits~\cite[Prop.~2.13.5]{Borceux1994} and one can check that the
  forgetful functor from pointed sets to sets creates limits (we
  adopt~\cite[Def.~3.3.1]{Riehl2016}), so that the forgetful functor \(U\)
  sending an abelian group to its underlying set pointed at the neutral element
  also preserves filtered colimits.
  Hence, \(UF X \cong \colim_{(n,f) : I_X}UF[n]\).

  Because each \([n]\) has decidable equality, we can directly check that the
  unit maps \([n] \to UF[n]\) are all injective.
  But filtered colimits commute with finite limits in sets (see
  \cite[Thm.~2.13.4]{Borceux1994} or \cite[Thm.~3.8.9]{Riehl2016}) and the
  forgetful functor from pointed sets to sets creates such colimits and limits,
  so they also commute in pointed sets.
  Thus, since monos can be characterized using pullbacks, the unit map
  \(X \to UFX\) must also be injective, as desired.
\end{proof}

\begin{proposition}\label{k-acyclic-stabilization}
  For natural numbers \(n \geq 1\) and \(k \geq 2\), the \(n\)-fold suspension
  \(\susp^n X\) of a type \(X\) is \(k\)-connected if and only if X is
  \((k - n + 1)\)-acyclic.
\end{proposition}
\begin{proof}
  If \(X\) is \((k - n + 1)\)-acyclic, then \(\susp X\) is
  \((k - n + 1)\)-connected, so that \(\susp^n X\) is
  \(k = (k - n + 1 + (n - 1))\)-connected by~\cite[Thm.~8.2.1]{HoTTBook}.

  For the converse, note that the \(n = 1\) case holds by definition.
  For \(n = 2\), suppose that \(\susp^2 X\) is \(k\)-connected. Then \(\susp X\)
  is \(k\)-acyclic, so by \cref{suspension-2-acyclic-iff-type-connected} and the
  fact that \(k \geq 2\), we see that \(X\) is connected.
  But then \(\susp X\) is simply connected and by \cref{k-acyclic-iff-connected}
  even \((k-1)\)-connected. Hence \(X\) is \(k - 1 = (k - 2 + 1)\)-acyclic, as
  we wished to show.

  Now suppose that \(n > 2\) and that \(\susp^n X\) is \(k\)-connected. Since
  \(\susp^n X \simeq \susp^2\susp^{n-2} X\), we see that \(\susp^{n-2}X\) is
  \(k - 2 + 1 = (k - 1)\)-acyclic, i.e., \(\susp^{n-1}X\) is
  \((k-1)\)-connected. So by induction hypothesis, \(X\) is
  \(k - 1 - (n - 1) + 1 = (k - n + 1)\)-acyclic, as desired.
\end{proof}

\section{The plus principle}\label{sec:plus}

From the definition of epimorphisms, we know that the type of extensions of a
map $f' : A \to X$ along an epimorphism $f: A \to B$,
\[
  \begin{tikzcd}
    A \ar[d,"f"']\ar[r,"f'"] & X \\
    B, \ar[ur,dashed,"h"']
  \end{tikzcd}
\]
is a proposition: indeed, it is the fiber at $f'$ of the precomposition
embedding $f^* : (B \to X) \hookrightarrow (A \to X)$.  It is then a natural
question to ask for equivalent reformulations of this proposition that might be
easier to check.  First we observe that a necessary condition is
$\ker(\pi_1(f)) \subseteq \ker(\pi_1(f'))$, either in the sense of inclusion
among congruence relations $\Trunc A_1\times\Trunc A_1 \to \Set$, or inclusions
of subgroups for each $a:A$.

We don't know whether this is sufficient in general. However, there is a
seemingly quite innocuous assumption under which it is, which we dub the
\emph{plus principle} (PP):
\begin{principle}[\flinkspec{synthetic-homotopy-theory}{plus-principle}{definition} PP]
  Every acyclic and simply connected type is contractible.
\end{principle}
Hoyois highlighted this in the context of Grothendieck
$(\infty,1)$-topoi~\cite[Rem.~4]{Hoyois2019}, and it seems to be open whether
it's true in general in that context.  It follows from \emph{Whitehead's
principle}~\cite[Sec.~8.8]{HoTTBook} (every infinitely connected type is
contractible, a.k.a.\ \emph{hypercompleteness}) by the following:
\begin{proposition}\label{lem:acyclic-simply-connected}
  Any acyclic and simply connected type is infinitely connected.
\end{proposition}
\begin{proof}
  This follows directly from~\cref{k-acyclic-iff-connected}.
\end{proof}

\begin{remark}[Anel]\label{Whitehead-plus-principles}
  While Whitehead's principle does \emph{not} hold in the \(\infty\)-topos of
  parametrized spectra (an object is hypercomplete if and only if the spectrum
  part is trivial), the plus principle \emph{does} hold there, as observed by
  Mathieu Anel (private communication).
  The outline of his argument is as follows: We write \(\mathsf S\),
  \(\mathsf {Sp}\) and \(\mathsf {PSp}\) for the \(\infty\)-categories of
  spaces, spectra, and parametrized spectra, respectively.
  The canonical functors \(\mathsf S \to \mathsf {PSp} \to \mathsf S\) are
  both left and right adjoint to each other and hence both preserve suspensions
  (as well as \(n\)-connected/truncated objects).
  The inclusion functor \(\mathsf {Sp} \to \mathsf {PSp}\) preserves weakly
  contractible colimits and hence suspensions.  Moreover, the suspension functor
  in \(\mathsf {Sp}\) is an equivalence.
  Now, if \(E\) is an object of \(\mathsf {PSp}\) and \(B\) is its image in
  \(\mathsf S\) by \(\mathsf{PSp} \to \mathsf S\) (its base), then \(B\) is
  respectively acyclic and simply-connected if \(E\) is.
  Thus, if \(E\) is acyclic and simply-connected, then \(B \simeq \One\), and
  thus \(E\) is a spectrum. But then \(\susp E \simeq \One\) implies that
  \(E \simeq \Zero\) as a spectrum (i.e., \(E\) is terminal in
  \(\mathsf {PSp}\)). So all acyclic simply connected objects in
  \(\mathsf {PSp}\) are terminal.
\end{remark}

From the plus principle itself we can deduce an analogous result for maps. We add
(PP) to indicate that the result assumes the plus principle.
\begin{lemma}[PP]
  Any acyclic $1$-equivalence is an equivalence.
\end{lemma}
\begin{proof}
  Consider an acyclic $1$-equivalence $f : A \to B$. We show that each fiber is
  contractible, so let $b:B$ be given, and let $F$ be the fiber of $f$ at $b$:
  \[
    \begin{tikzcd}
      F \ar[d]\ar[r,"g"]\pbcorner & A \ar[d,"f"] \\
      \One \ar[r,"b"'] & B
    \end{tikzcd}
    \hspace{-1cm}
    \begin{tikzcd}[column sep=6mm]
      & \dots \ar[r] & \pi_2(B) \ar[out=-20,in=160,"\delta"']{dll} \\
      \pi_1(F) \ar[r,"\pi_1(g)"]
      & \pi_1(A) \ar[r,"\pi_1(f)"] & \pi_1(B) \ar[out=-20,in=160]{dll} \\
      \pi_0(F) = \One
    \end{tikzcd}
  \]
  From the displayed fragment of the long exact sequence~\cite[Sec.~8.4]{HoTTBook}
  (relative to any base point of $F$) we have
  $\im(\delta) = \ker(\pi_1(g)) = \pi_1(F)$, since $\pi_1(f)$ is an
  isomorphism. Thus, $\pi_1(F)$ is abelian as well as perfect
  (by~\cref{2-acyclic-type-char}) and hence trivial.  By (PP) it follows
  that $F$ is contractible.
\end{proof}

We show that the necessary condition identified above is sufficient under
(PP):
\begin{proposition}[PP]\label{acyclic-extension}
  Let $f : A \to B$ be acyclic and ${f' : A \to X}$ any map. Then $f'$ extends
  along $f$ if and only if we have the inclusion
  $\ker(\pi_1(f)) \subseteq \ker(\pi_1(f'))$.
\end{proposition}
\begin{proof}
  The condition is necessary by functoriality of $\pi_1$, since any extension
  $h$ satisfies $f' = h\circ f$.

  To establish sufficiency, we note that it suffices to consider the case where
  $f'$ is surjective, since otherwise we just extend the corestriction to the
  image of $f'$. Now form the pushout:
  \[
    \begin{tikzcd}
      A \ar[r,"f'"]\ar[d,"f"']\pocorner & X \ar[d,"g"] \\
      B \ar[r,"g'"'] & P
    \end{tikzcd}
  \]
  Then $g$ is acyclic and surjective on $\pi_1$. By the previous lemma, it
  suffices to show that $g$ is also injective on $\pi_1$. Picking a base point
  in $A$, making the whole square pointed, we get a pushout square in groups by
  the Seifert--van~Kampen theorem~\cite[Thm.~8.7.12]{HoTTBook}:
  \[
    \begin{tikzcd}
      \pi_1(A) \ar[r,"\pi_1(f')"]\ar[d,"\pi_1(f)"']\pocorner & \pi_1(X) \ar[d,"\pi_1(g)"]\ar[ddr,bend left,"\id"] & \\
      \pi_1(B) \ar[r,"\pi_1(g')"']\ar[drr,dotted,bend right=15,"\varphi"'] & \pi_1(P)\ar[dr,dashed,"\psi"] & \\
      & & \pi_1(X)
    \end{tikzcd}
  \]
  The inclusion
  $\ker(\pi_1(f)) \subseteq \ker(\pi_1(f'))$ yields the dotted map \(\varphi\) as
  \[
    \pi_1(B) \simeq \pi_1(A)/\ker(\pi_1(f)) \to \pi_1(A)/\ker(\pi_1(f')) \to
    \pi_1(X).
  \]
  This induces the dashed map $\psi$, a retraction of $\pi_1(g)$.
\end{proof}

The above result was established for CW-spaces in
\cite[Prop.~3.1]{HausmannHusemoller1979} and by different means for path
connected CW-spaces in \cite[Cor.~4.4]{Alonso1983}

\begin{corollary}[PP]\label{plus-unique}
  Let $A$ be a pointed, connected type with a given perfect normal subgroup
  $P \trianglelefteq \pi_1(A)$.  Then the type of acyclic maps $f : A \to X$
  with $\ker(\pi_1(f)) = P$ is a proposition.
\end{corollary}

Another application of the plus principle is the following stabilization result.

\begin{proposition}[PP]\label{acyclic-stabilization}
  For any natural number \(n \geq 1\), the \(n\)-fold suspension \(\susp^n X\)
  of a type \(X\) is contractible if and only if \(X\) is acyclic.
\end{proposition}
\begin{proof}
  The right-to-left implication is immediate. We prove the converse by
  induction. For \(n = 1\) it follows by definition. For \(n = 2\), we assume
  that \(\susp^2 X\) is contractible so that \(\susp X\) is acyclic.
  Then \(X\) must be connected by \cref{suspension-2-acyclic-iff-type-connected}
  so that \(\susp X\) is simply connected and hence contractible by the plus
  principle.
  For \(n > 2\), we assume that \(\susp^n X\) is contractible. By the above,
  \(\susp^{n-2} X\) must be acyclic and hence \(\susp^{n-1} X\) is contractible,
  so that \(X\) is acyclic by induction hypothesis.
\end{proof}

Although it seems plausible, we do not know whether, in the absence of
Whitehead's Principle, a type \(X\) is acyclic as soon as its suspension
spectrum \(\susp^\infty X\) is contractible.

\subsection{The Blakers--Massey theorem for acyclic maps}
\label{sec:blakers-massey}

The Blakers--Massey theorem holds in HoTT for the left class of any
modality~\cite[Thm.~4.2]{ABFJ2020}.

Here we show directly that it holds for acyclic maps, assuming the plus
principle.  First we need the following lemma, also observed by
Raptis~\cite[Lem.~3.6]{Raptis2019} with a different proof.

\begin{lemma}[PP]\label{inh-join-acyclic-iff-contractible}
  If $A$ and $B$ are inhabited, then the type $A \join B$ is contractible if and
  only if it is acyclic.
\end{lemma}
\begin{proof}
  For the nontrivial direction, assume $A \join B$ is acyclic.  Since $A$ and
  $B$ are inhabited and we are proving a proposition, we may assume they are
  pointed.  Then $A \join B$ is equivalent to $\susp(A \smashpr B)$ by
  \cref{join-is-suspension-of-smash}. Since \(A\) and \(B\) are
  \((-1)\)-connected, \cite[Prop.~4.3.1]{Brunerie2016} tells us that
  \(A \smashpr B\) is \(0\)-connected, so that its suspension is
  \(1\)-connected by \cite[Thm.~8.2.1]{HoTTBook}.
  But now \(A \join B\) is acyclic and \(1\)-connected, hence contractible by
  the plus principle.
\end{proof}

In addition, we shall need the following observations,
giving a constructive treatment of \cite[Prop.~3.7]{Raptis2019}.
\begin{lemma}\label{join-inhabited}
  If the join \(A \join B\) is inhabited, then either \(A\) or \(B\)
  is inhabited.
\end{lemma}
\begin{proof}
  Since we are proving a proposition, we may assume we have an element
  of the join \(A \join B\).
  By join induction, we get two (point constructor)
  cases, so the conclusion follows.
\end{proof}

\begin{lemma}[PP]\label{join-loops-acyclic}
  Let $A$ be a pointed type and $B$ be any type.
  If the type $\Omega A \join B$ is acyclic,
  then it is contractible.
\end{lemma}
\begin{proof}
  The type $\Omega A \join B$, pointed at $\inl(\refl)$,
  is $0$-connected with perfect fundamental group by~\cref{2-acyclic-type-char},
  so it suffices to show that $\pi_1(\Omega A\join B)$ is abelian.
  By the naive van~Kampen theorem~\cite[Thm.~8.7.4]{HoTTBook},
  we can express the fundamental group as a set quotient of the type of sequences
  \[
    (\refl,\alpha_0,p_1,b_1,r_1,b_1',p_1',\alpha_1,p_2,\dots,b_n',p_n',\alpha_n,\refl)
  \]
  where
  \begin{itemize}
  \item $n:\mathbb N$,
  \item $p_k,p_k' : \Omega A$, $b_k,b_k': B$, for $0<k\le n$,
  \item $\alpha_k : p_k' = p_{k+1}$ for $0\le k\le n$
    with $p_0' \colonequiv p_{n+1}\colonequiv\refl$,
  \item $r_k : b_k = b_k'$ for $0<k\le n$.
  \end{itemize}
  To prove the proposition that two such codes give commuting elements,
  we look whether any of them has $n>0$. If so, we know $B$ is inhabited,
  and then $\Omega A\join B$ is contractible by~\cref{inh-join-acyclic-iff-contractible}.
  Otherwise, the two codes represent $2$-loops $\alpha,\beta:\Omega^2 A$,
  which commute by the Eckmann--Hilton argument.
\end{proof}

\begin{lemma}[PP]\label{join-ids-acyclic}
  For any types $A$ and $B$ with elements $a,a':A$ and $b,b':B$
  we have that the join $(a=_Aa') \join (b=_Bb')$ is contractible
  if and only if it acyclic.
\end{lemma}
\begin{proof}
  Suppose the join is acyclic.
  By~\cref{join-inhabited}, one join summand is inhabited,
  so without loss of generality,
  we may assume we have $p:a=_Aa'$.
  Concatenating with the inverse of $p$ gives an equivalence
  $(a =_A a') \simeq \Omega(A,a)$.
  Now apply~\cref{join-loops-acyclic}.
\end{proof}

Now we can prove the Blakers--Massey theorem for acyclic maps.
As observed by Raptis~\cite[Sec.~3.3]{Raptis2019},
the conclusion is slightly stronger than naively expected.
\begin{theorem}[PP, Blakers--Massey for acyclic maps]\label{BM-acyclic}
  Consider a pushout square:
  \[
    \begin{tikzcd}
      A \ar[r,"g"]\ar[d,"f"']\pocorner & C \ar[d] \\
      B \ar[r] & P
    \end{tikzcd}
  \]
  Then the relative pushout product
  $\Delta f \popr_A \Delta g$
  is acyclic if and only if it is contractible,
  and in that case the square is cartesian.
  This holds also if the absolute pushout product
  $\Delta f \popr \Delta g$ is acyclic.
\end{theorem}
\begin{proof}
  Assume that the relative pushout product is acyclic.
  The fibers of $\Delta f\popr_A \Delta g$ are joins of the form
  \[
    \bigl((a,\refl) =_F (a',p)\bigr)
    \join
    \bigl((a,\refl) =_G (a'',q)\bigr)
  \]
  for $a,a',a'':A$, $p:f(a)=f(a')$, and $q:g(a)=g(a'')$,
  where $F \colonequiv \fib_f(f(a))$ and
  $G \colonequiv \fib_g(g(a))$.
  Now apply~\cref{join-ids-acyclic} to get that the fibers are contractible.
  Then the little Blaker--Massey theorem gives that the square is cartesian.
  The final remark follows from the fact that $\Delta f \popr_A \Delta g$
  is a pullback of ${\Delta f \popr_A \Delta g}$.
\end{proof}

\section{Hypoabelian types and orthogonality}\label{sec:hypoab}

In the context of higher topos theory, Hoyois showed that the acyclic maps are
part of an orthogonal factorization system~\cite{Hoyois2019}.
While we leave a type-theoretic construction of this factorization system to
future work (see~\cref{sec:conclusion}), we consider what the corresponding
right class should be, namely that of \emph{hypoabelian} maps.

\begin{definition}[Hypoabelianness]\label{def:hypoabelian}
  A type \(X\) is \emph{hypoabelian} if every perfect subgroup of \(\pi_1(X,x)\)
  is trivial, for every \(x:X\).  A map \(f : X \to Y\) is hypoabelian if all
  its fibers are.
\end{definition}
We note that a type \(X\) is hypoabelian if and only if its
\(1\)\nobreakdash-truncation is.  An equivalent definition says that the perfect
core (i.e., the largest perfect subgroup) of each \(\pi_1(X,x)\) is trivial.
We also remark that this definition, at least in the absence of propositional
resizing~\cite[Sec.~3.5]{HoTTBook}, should be understood as being relative to a
type universe.

Recall that a map $f:A \to B$ is \emph{left orthogonal} to a map $g:X \to Y$,
denoted $f\perp g$, if we have a contractible type of lifts for all squares as
below left, or equivalently, if the square below right is a pullback square.
\[
  \begin{tikzcd}
    A \ar[r]\ar[d,"f"'] & X\ar[d,"g"] \\
    B \ar[r]\ar[ur,dashed] & Y
  \end{tikzcd}
  \qquad
  \begin{tikzcd}
    X^B \ar[d,"g_*"']\ar[r,"f^*"] & X^A\ar[d,"g_*"]\\
    Y^B \ar[r,"f^*"'] & Y^A
  \end{tikzcd}
\]
In case of maps to the terminal type $\One$, we write $A\perp X$, and say that
$A$ is left orthogonal to $X$.  The is equivalent to the constants
map $\const:X \to (A \to X)$ being an equivalence.
\begin{lemma}[PP]
  For all acyclic types $A$ and hypoabelian types $X$, we have $A \perp X$.
\end{lemma}
\begin{proof}
  We need to show that the type of extensions of a map $g : A \to X$ along the
  terminal map $A \to \One$ is contractible. This is a proposition since $A$ is
  acyclic, so it suffices to check that there exists an extension.  Picking a
  base point of $A$, it suffices by~\cref{acyclic-extension} to check that
  $\pi_1(A)=\ker(\pi_1(g))$, or equivalently, that $\im(\pi_1(g))$ is
  trivial. This follows since $\pi_1(A)$ is perfect, and the fact that the image
  of a perfect group is perfect.
\end{proof}
\begin{corollary}[PP]
  For all acyclic maps $f : A \to B$ and hypoabelian maps $g : X \to Y$, we have
  $f\perp g$.
\end{corollary}
\begin{proof}
  This is a general fact about two classes $\mathcal{L}$ and $\mathcal{R}$ of
  maps defined in terms of fibers, i.e., a map is in $\mathcal{L}/\mathcal{R}$
  if and only if all its fibers are.  Suppose we have orthogonality of terminal
  maps in $\mathcal{L}$ against terminal maps in $\mathcal{R}$.  Then we get
  $f \perp g$ for all $f\in\mathcal{L}$ and ${g\in\mathcal{R}}$.  Indeed,
  expressing a lifting problem in terms of a map ${\varphi : B \to Y}$ and a
  fiberwise map ${\psi : \prod_{b:B}A(b) \to X(\varphi(b))}$,
  \[
    \begin{tikzcd}
      \sum_{b:B}A(b) \ar[r,"\tilde\psi"]\ar[d] & \sum_{y:Y}X(y) \ar[d] \\
      B\ar[r,"\varphi"']\ar[ur,dashed] & Y
    \end{tikzcd}
  \]
  the type of lifts is
  \begin{align*}
      \sum_{h : \prod_{b:B}X(\varphi(b))}\prod_{b:B}\prod_{a:A(b)}\psi_b(a)
    = h(b)
    &\,\simeq\,
      \prod_{b:B}\sum_{x:X(\varphi(b))}\prod_{a:A(b)}\psi_b(a)=x \\
    &\,\simeq\, \prod_{b:B} \fib_{\const}(\psi_b).
  \end{align*}
  So if all the constants maps,
  $X(\varphi(b)) \to (A(b) \to X(\varphi(b)))$, are equivalences, then
  all the fibers, $\fib_{\const}(\psi_b)$, are contractible, so this type of lifts
  is contractible too.
\end{proof}

Nilpotent types~\cite{Scoccola2020} are a special case of hypoabelian types.
We will show (\cref{acyclic-nilpotent-orth}) that every nilpotent type that
is the limit of its Postnikov tower is right orthogonal to acyclic types
\emph{without} assuming the plus principle.
(In the classical model, every type is the limit of its Postnikov tower.)

Recall that in homotopy type theory, the Postnikov tower of a type \(X\) is
given by the truncation maps
\[
  X \to \dots \to \squash{X}_n \to \dots \to \squash{X}_1 \to \squash{X}_0.
\]

Following~\cite[Def.~7.2.2.20]{LurieHTT}, we define a type $Y$ to be an \emph{EM
  $n$-gerbe} if it is $(n-1)$-connected and $n$\nobreakdash-truncated.
If $n\ge 2$, then this determines an abelian group ${H := \pi_n(Y,y)}$,
which doesn't depend essentially on $y:Y$.
If $n=1$, we additionally require that $\pi_1(Y,y)$ is abelian
for any/all $y:Y$.
Now, any map ${P : A \to K(H,n+1)}$ determines a family of $n$-gerbes
over~$A$ via the equivalence $K(H,n+1) \simeq {\sum_{Z:\mathcal{U}}\Trunc{Z=K(H,n)}_0}$.
We call such a $P$ a \emph{principal EM fibration}.
(See also~\cite{BCFR}, the resulting gerbes are \emph{banded} by $H$.)

\begin{lemma}\label{lem:acyclic-banded-gerbe-orth}
  Given a pointed acyclic type $A$, and a pointed
  EM $n$-gerbe $Y$, we have that $(A \ptdto Y)$ is contractible.
\end{lemma}
\begin{proof}
  Since \(Y\) is a pointed \(n\)-gerbe, we may assume
  $Y \simeq K(H,n) \simeq \Omega K(H,n+1)$, where $H$ is the associated abelian
  group. Now, using~\cite[Lem.~6.5.4]{HoTTBook}, we have
  \begin{align*}
    (A \ptdto Y) &\simeq (A \ptdto \Omega K(H,n+1)) \\
    &\simeq (\susp A \ptdto K(H,n+1)) \\
    &\simeq (\One \ptdto K(H,n+1)) \simeq \One.\qedhere
  \end{align*}
\end{proof}

\begin{proposition}\label{acyclic-nilpotent}
  Given a pointed acyclic type $A$ and a pointed nilpotent type $X$, we have
  that $(A \ptdto X)$ is contractible if in addition $X$ is the limit of its
  Postnikov tower.
\end{proposition}

To prove~\cref{acyclic-nilpotent}, we recall
from~\cite[Thm.~2.58]{Scoccola2020} that $X$ is nilpotent if and only if each
map $\Trunc X_{n+1} \to \Trunc X_n$ in the Postnikov tower of $X$ factors as a
finite composition of principal EM fibrations, i.e., maps
classified by $K(A,n+1)$ for abelian groups $A$.
\begin{proof}
  Since $X \simeq \varprojlim_n \Trunc X_n$, we get an equivalence between
  $(A \ptdto X)$ and $\varprojlim_n (A \ptdto \Trunc X_n)$, so it suffices to
  show that the type $(A \ptdto \Trunc X_n)$ is contractible for all $n$, which
  we do by induction on $n$.

  In the step case, since we're proving a proposition, we may assume that the
  map $\Trunc X_{n+1} \to \Trunc X_n$ is factored as
  \[
    \Trunc X_{n+1} = Y_k \to Y_{k-1} \to \dots \to Y_0 = \Trunc X_n,
  \]
  with each map a principal EM fibration.  The result now follows
  from~\cref{lem:acyclic-banded-gerbe-orth}.
\end{proof}

\begin{corollary}\label{acyclic-nilpotent-orth}
  For all acyclic types $A$ and nilpotent types~$X$ that are limits of their
  Postnikov towers, we have $A \perp X$.
\end{corollary}
\begin{proof}
  We're proving a proposition, so fix a base point $\pt:A$.  The evaluation map
  at $\pt$ fits in the diagram below:
  \[
    \begin{tikzcd}
      X \ar[dr,"\id"']\ar[r,"\const"] & (A \to X)\ar[d,"\ev_\pt"] \\
      & X
    \end{tikzcd}
  \]
  By $3$-for-$2$ for equivalences, it suffices to show that $\ev_\pt$ is an
  equivalence.  For each $x_0:X$, the fiber is the type of pointed maps
  $(A \ptdto X)$, where $X$ is pointed at $x_0$. And this is contractible by the
  previous proposition.
\end{proof}

\section{Acyclicity via (co)homology}\label{sec:homology}

Classically, the acyclic types are characterized as types $A$ whose reduced
integral homology vanishes, i.e., $\tilde\upH_i(A)=0$ for all $i$.  This is in
fact the origin of the name \emph{acyclic}, meaning that every cycle is a
boundary, using the chain complex model of homology.
For a definition of reduced homology in homotopy type theory,
see~\cite[Def.~3.10]{ChristensenScoccola2023}. Up to equivalence,
$\tilde\upH_i(X)$ is defined for pointed types $X$
as $\pi_i(\upH\mathbb{Z} \smashpr \susp^\infty X)$.
For an unpointed type $A$, we define \emph{unreduced homology} as
$\upH_i(A) \colonequiv \tilde\upH_i(A_+)$, where
$A_+$ is the free pointed type on $A$, viz.,
$A$ with a disjoint base point.


\begin{definition}
  A type \(A\) is \emph{homologically \(k\)-acyclic}
  if \(A\) is inhabited,
  and any one of the following equivalent conditions hold:
  \begin{enumerate}[label=(\roman*)]
  \item We have \(\tilde\upH_i(A) = 0\) for \(i\le k\)
    and any choice of base point.
  \item The map \(A \to \One\) induces isomorphisms
    \(\upH_i(A) \to \upH_i(\One)\) for \(i\le k\).
  \end{enumerate}
  A type is \emph{homologically acyclic}
  if it is homologically \(k\)-acyclic
  for all \(k\).
\end{definition}
Note that the augmentation map \(\upH_0(A) \to \Z\) is an equivalence
if and only if \(A\) is connected,
and if \(A\) is pointed, then
the inclusion \(A \to A_+\) gives equivalences
\(\tilde\upH_i(A) \simeq \upH_i(A)\) for \(i>0\).
This establishes the equivalence of the two conditions.

\begin{proposition}\label{prop:k-acyclic-iff-homologically-so}
  A type is $k$-acyclic if and only if
  it is homologically $k$-acyclic.
\end{proposition}
\begin{proof}
  Fix a type \(A\).
  For \(k=0\) both conditions amount to \(A\) being connected.
  For \(k>0\) we get from the suspension property of homology that
  \(\tilde\upH_i(\susp A) \simeq \tilde\upH_{i-1}(A)\) for all \(i\).
  Since \(\susp A\) is simply connected, the truncated Whitehead's
  theorem~\cite[Thm~8.8.3]{HoTTBook} implies that $\Trunc{\susp A}_{k+1}$ is
  contractible if and only if $\pi_i(\susp A)$ vanishes for $i\le k+1$. By
  Hurewicz' theorem~\cite[Prop.~3.17]{ChristensenScoccola2023}, this happens if and only if
  $\tilde\upH_i(\susp A)$ vanishes in the same range.
\end{proof}
\begin{corollary}
  Any acyclic type is homologically acyclic.
\end{corollary}

The converse holds assuming Whitehead's principle (WP).
\begin{corollary}[WP]
  Any homologically acyclic type is acyclic.
\end{corollary}

We define the notion of being \emph{cohomologically acyclic}
in analogy with being homologically acyclic,
just using integral cohomology~\cite{Cavallo2015,BLM2022} instead of integral homology.
Then we have the following.
\begin{lemma}
  Any homologically acyclic type is cohomologically acyclic.
\end{lemma}
\begin{proof}
  Fix a type \(A\).
  By the suspension property of cohomology,
  we may assume that \(A\) is pointed and \(1\)-connected.
  (Otherwise, consider the suspension \(\susp A\).)
  Thus, by~\cref{k-acyclic-iff-connected,%
  prop:k-acyclic-iff-homologically-so},
  \(A\) is in fact \(k\)-connected for any \(k\),
  but then
  \[
    \tilde\upH^i(A) \equiv \Trunc{A \ptdto \upK(\Z,i)}_0
    \simeq \Trunc{\One \ptdto \upK(\Z,i)}_0 \simeq 0,
  \]
  as desired.
\end{proof}
Note that it is more subtle to characterize $k$-acyclicity cohomologically,
since there are types \(A\) with \(\tilde\upH^i(A)=0\) for \(i\le k\)
that are not \(k\)-acyclic.
Consider for example \(A = \upK(\Z/2\Z,2)\), which has
\[
  \tilde\upH^2(A) = \Trunc{A \ptdto \upK(\Z,2)}_0
  = \Hom(\Z/2\Z, \Z) = 0,
\]
but \(\tilde\upH_2(A) = \pi_2(A) = \Z/2\Z\).
This can classically be fixed by requiring the induced map
\(\tilde\upH^{k+1}(A) \to \Hom(\tilde\upH_{k+1}(A),\Z)\)
to be an isomorphism.
However, this characterization relies on the Universal Coefficient Theorem
which is not expected to hold in HoTT
due to the presence of higher Ext groups~\parencite{ChristensenFlatenExt}.
Rather, we expect there is a Universal Coefficient Spectral Sequence as
in~\textcite[(UCT2)]{Adams1969}.
Even assuming this, we would still need some argument to infer that cohomological
acyclicity implies that the homology groups are finitely presented.
However, the traditional proofs of this are very classical~\parencite[Prop.~3F.12]{HatcherAT}.

Let us now move on to the homological characterization of acyclicity of maps.
Here it is not sufficient to just consider integer coefficients.
But we can always move to a universal cover, by virtue of the following observation.
\begin{lemma}\label{lem:acyclicuniversal}
  A map \(f : A \to B\) is acyclic if and only if,
  for all \(b : B\), the pullback of \(f\) to
  the \(1\)-connected cover of \(B\) at \(b\),
  is acyclic,
  \begin{equation}\label{eq:universalcover}
    \begin{tikzcd}
      \tilde A \ar[d]\ar[r,"f'"]\pbcorner & \tilde B_b \ar[d] \\
      A \ar[r,"f"'] & B,
    \end{tikzcd}
  \end{equation}
  where \(\tilde B_b\colonequiv \sum_{y:B}\Trunc{b=y}_0\).
\end{lemma}
\begin{proof}
  The fibers of \(f\) and \(f'\) are identified.
\end{proof}
We also need the following result, which would follow from a
Universal Coefficient Spectral Sequence for homology~\parencite[(UCT1)]{Adams1969}.
However, it also has a direct proof.
\begin{lemma}\label{lem:homology-acyclic-coeff}
  If a type \(A\) is homologically \(k\)-acyclic,
  then it is so for any abelian coefficient group \(L\):
  The map \(A \to \One\) induces isomorphisms
  \(\upH_i(A; L) \to \upH_i(\One; L)\) for \(i \le k\).
\end{lemma}
\begin{proof}
  It suffices to consider the case where \(A\) is connected,
  and then we assume it is pointed and consider reduced homology.
  Again, by suspending and shifting if necessary, we may assume
  \(A\) is simply connected, so being (homologically) \(k\)-acyclic
  amounts to being \(k\)-connected.
  Now conclude by~\textcite[Prop.~3.19]{ChristensenScoccola2023}.
\end{proof}

With these preliminaries, we are ready to present the following
definition, which refines~\textcite[Def.~1.2]{HausmannHusemoller1979}
by considering \(k\)-acyclicity instead of acyclicity \emph{simpliciter}:
\begin{definition}
  A map \(f : A \to B\) is \emph{homologically \(k\)-acyclic} if any one
  of the following equivalent conditions hold:
  \begin{enumerate}[label=(\roman*)]
  \item\label{it-HH-i}
    All fibers of \(f\) are homologically \(k\)-acyclic.
  \item\label{it-HH-ii}
    For any local coefficient system \(L : B \to \mathrm{AbGroup}\),
    the induced maps
    \[
      f_* : \upH_i(A; f^*L) \to \upH_i(B; L)
    \]
    are isomorphisms for \(i \le k\) and surjective for \(i = k+1\).
  \item\label{it-HH-iii}
    The induced maps
    \[
      f_* : \upH_i(A; f^*\Z\pi_1B) \to \upH_i(B; \Z\pi_1B)
    \]
    are isomorphisms for \(i \le k\) and surjective for \(i = k+1\).
  \item\label{it-HH-iv}
    For each \(b:B\), the map \(f'\) as in \eqref{eq:universalcover}
    induces isomorphisms \(\upH_i(\tilde A) \to \upH_i(\tilde B_b)\)
    for \(i \le k\) and a surjection for \(i = k+1\).
  \end{enumerate}
\end{definition}
\begin{proof}[Proof of the equivalence]
  For \labelcref{it-HH-i} implies \labelcref{it-HH-ii}: We use the Serre Spectral Sequence
  for homology, as developed by~\textcite[Sec.~5.5]{vanDoorn2018}:
  \[
    E^2_{p,q} = \upH_p(B; \lambda b. \upH_q(F(b); L(b))) \Rightarrow
    \upH_{p+q}(A; f^*L),
  \]
  where \(F(b)\) is the fiber of \(f\) at \(b\).  By \labelcref{it-HH-i} and
  \cref{lem:homology-acyclic-coeff}, we have \({\upH_q(F(b); L(b)) = 0}\) for
  \(0<q\le k\), so the first possibly nontrivial differential (by total degree)
  is the transgression
  \[
    d^{k+2}_{k+2,0} : \upH_{k+2}(B; L) \to
    \upH_0(B; \lambda b. \upH_{k+1}(F(b); L(b))).
  \]
  Thus, convergence immediately gives isomorphisms
  \(\upH_i(A; f^*L) \to \upH_i(B; L)\) for \(i \le k\)
  and a short exact sequence
  \[ 0 \to E^\infty_{0,k+1}
    \to \upH_{k+1}(A; f^*L) \to \upH_{k+1}(B; L) \to 0,\]
  yielding \labelcref{it-HH-ii}.

  It is clear that \labelcref{it-HH-ii} implies \labelcref{it-HH-iii}.

  For \labelcref{it-HH-iii} implies \labelcref{it-HH-iv}, we look at the map of
  fibrations induced by a horizontal reading of~\eqref{eq:universalcover}, with
  fibers \(\pi_1(B,b)\).  The implications follow from naturality of the Serre
  Spectral Sequences and the Five Lemma.

  For \labelcref{it-HH-iv} implies \labelcref{it-HH-i}, we fix
  \(b:B\) and use~\cref{lem:acyclicuniversal} to get a fiber sequence
  \(F(b) \to \tilde A \to \tilde B_b\).
  For notational simplicity, we may as well assume that \(B\) is already simply connected,
  as well as pointed at \(b:B\).
  Then we can use the Serre Spectral Sequence
  for homology with constant integral coefficients, in particular, we look at
  naturality with respect to the map of fiber sequences:
  \[
    \begin{tikzcd}
      F(b) \ar[r]\ar[d] & A\ar[r,"f"]\ar[d,"f"'] & B\ar[d,"\id"] \\
      \One \ar[r] & B\ar[r,"\id"'] & B
    \end{tikzcd}
  \]
  The comparison maps the left edge of the domain,
  \(E^2_{0,i} = \upH_0(B;\upH_i(F(b))) = \upH_i(F(b))\),
  to the left edge of the codomain,
  \(E^2_{0,i} = \upH_0(B;\upH_i(\One)) = \upH_i(\One)\),
  yielding \labelcref{it-HH-i}, as desired.
\end{proof}

\begin{corollary}
  A map is \(k\)-acyclic if and only if it is homologically \(k\)-acyclic.
  Any acyclic map is homologically acyclic, and the converse follows from \textup{(WP)}.
\end{corollary}

We can also make the analogous definition for being a cohomologically acyclic map
using the Serre Spectral Sequence for cohomology,
which has even been formalized~\parencite[Sec.~2.3]{vanDoorn2018}.
But again we only know that homologically acyclic maps are cohomologically acyclic,
as the converse would again require a Universal Coefficient Spectral Sequence and
an argument to ensure finitely presented homology groups.

\section{Examples of acyclic types}\label{sec:examples-acyclic}
We finally give some nontrivial examples of acyclic types.

\subsection{A 2-dimensional acyclic type}
\label{sec:hatcher}
Our first example is Hatcher's $2$-dimensional
complex~\cite[Ex.~2.38]{HatcherAT}.  We import this as the higher inductive type
(HIT) $X$ with constructors:
\[
  \pt : X,\quad a,b : \Omega X, \quad r : a^5=b^3, \quad s : b^3=(ab)^2
\]

\begin{definition}[\flinkspec{synthetic-homotopy-theory}{hatchers-acyclic-type}{algebras-with-the-structure-of-hatchers-acyclic-type} %
  Hatcher structure and algebra]
  A \emph{Hatcher structure} on a pointed type \(A\) is given by identifications
  \[
    a,b : \Omega A, \quad r : a^5=b^3, \quad s : b^3=(ab)^2.
  \]
  A \emph{Hatcher algebra} is a pointed type equipped with Hatcher structure.
\end{definition}

The HIT \(X\) is precisely the \emph{initial} Hatcher algebra.

\begin{lemma}[\flinkspec{synthetic-homotopy-theory}{hatchers-acyclic-type}{loop-spaces-uniquely-have-the-structure-of-a-hatcher-acyclic-type}]%
  \label{unique-Hatcher-structure}
  Every loop space, pointed at \(\refl\), has a unique Hatcher structure.
\end{lemma}
\begin{proof}
  The type of Hatcher structures on a loop space \(\Omega A\) is
  \[
    \sum_{a,b : \Omega^2 A}\pa*{a^5=b^3} \times \pa*{b^3 = (ab)^2}.
  \]
  By Eckmann-Hilton~\cite[Thm~2.1.6]{HoTTBook}, we have \(ab = ba\), so the last
  component is equivalent to \(b = a^2\), and can be contracted away to
  obtain: \(\sum_{a: \Omega^2 A}\pa*{a^5=a^6}\).
  But, cancelling \(a^5\), this is equivalent to the contractible type
  \(\sum_{a : \Omega^2 A} \pa*{a = \refl}\).
\end{proof}

\begin{proposition}[\flinkspec{synthetic-homotopy-theory}{hatchers-acyclic-type}{hatchers-acyclic-type-is-acyclic}]%
  \label{Hatcher-acyclic}
  The type $X$ is acyclic.
\end{proposition}
\begin{proof}
  For all pointed types \(Y\), we have:
  \begin{align*}
    (\susp X \ptdto Y)
    &\simeq (X \ptdto \Omega Y) \\
    &\simeq \HatcherStr(\Omega Y) \\
    &\simeq \One,
  \end{align*}
  where the first equivalence is~\cite[Lem.~6.5.4]{HoTTBook}, the second is the
  universal property of \(X\), and the third is~\cref{unique-Hatcher-structure}.

  Thus, \(\susp X\) has the universal property of the unit type and hence must be
  contractible.
\end{proof}

In his lecture notes on higher topos theory, Charles Rezk
asked~\cite[p.~11]{Rezk2019} whether it is possible to give a purely
type-theoretic proof of the fact that \(X \to \One\) is an epimorphism.
Together with our characterization of the epimorphisms as the acyclic maps,
\cref{Hatcher-acyclic} positively answers Rezk's question.

The nontriviality of the type \(X\) follows from the following result since any
\(0\)-connected map \(X \to \BA_5\) gives a surjection \(\pi_1(X) \to \AG_5\)
by~\cite[Cor.~8.4.8]{HoTTBook}.

\begin{proposition}
  The type $X$ has a $0$-connected map to $\BA_5$, the classifying type of the
  alternating group $\AG_5$.
\end{proposition}
\begin{proof}
  Let $a = (1\,2\,3\,4\,5)$ and $b=(2\,5\,4)$ in $\AG_5$. We have
  $ab = (1\,2)(3\,4)$, so these satisfy $a^5=b^3$ and
  $b^3=(ab)^2$, and thus induce a group homomorphism $\pi_1(X) \to \AG_5$,
  corresponding to a (pointed) map ${X \to \Trunc X_1 \to \BA_5}$.  Since $a$
  and~$b$ generate $\AG_5$, the group homomorphism is surjective, so the map
  $X \to \BA_5$ is $0$-connected by~\cite[Cor.~8.8.5]{HoTTBook}.
\end{proof}

\subsection{Higman's type}
\label{sec:higman}

Another interesting example of an acyclic type is the classifying type of
Higman's group $\HG$~\cite{Higman1951} which is given by the presentation
\[
  \HG = \langle\,a,b,c,d \mid a=[d,a], b=[a,b], c=[b,c], d=[c,d]\,\rangle.
\]
(Here, $[x,y]$ denotes the commutator $[x,y]=xyx^{-1}y^{-1}$.)  We will show
that $\BH$ is acyclic, and moreover, that this presentation is
\emph{aspherical}, meaning that the presentation complex is already a $1$-type,
see also~\cite{DyerVasquez1973}.  The presentation complex is easily imported
into HoTT as the HIT $\BH$ with a point constructor $\pt:\BH$, four path
constructors $a,b,c,d:\Omega\BH$, and four $2$-cell constructors corresponding
to the relations.

\begin{proposition}
  The type $\BH$ is acyclic.
\end{proposition}
\begin{proof}
  By Eckmann--Hilton, similarly
  to~\cref{unique-Hatcher-structure} and \cref{Hatcher-acyclic}.
\end{proof}

To show that $\BH$ is not contractible, we make use of the following
result due to David W\"arn~\cite[Lem.~8, Thm.~9]{Warn2023}.

\begin{theorem}[Wärn]\label{thm:warn}
  Given a span $A \leftarrow R\to B$ of $0$\nobreakdash-truncated maps of
  $1$-types, its pushout $A+_RB$ is a $1$-type, the inclusion maps are
  $0$-truncated, and the gap map is an embedding.
\end{theorem}
\begin{theorem}\label{Higman-is-nontrivial}
  The type $\BH$ is a $1$-type, and the generators $a,b,c,d$ have infinite
  order.
\end{theorem}
\begin{proof}
  Indeed, $\BH$ can be re-expressed as an iterated pushout as follows:
  \begin{equation}\label{eq:higman-decomp}
    \begin{tikzcd}[column sep=6mm, row sep=6mm]
      \Bgen b\ar[r]\ar[d]\pocorner & \Bgen{b,c}\ar[d] \\
      \Bgen{a,b}\ar[r] & \Bgen{a,b,c}
    \end{tikzcd}
    \quad
    \begin{tikzcd}[column sep=6mm, row sep=6mm]
      \Bgen{a,c}\ar[r]\ar[d]\pocorner & \Bgen{a,b,c}\ar[d] \\
      \Bgen{c,d,a}\ar[r] & \BH
    \end{tikzcd}
  \end{equation}
  Here, each type is the HIT that uses only the constructors of $\BH$ that
  involve the mentioned generators.  In particular, $\Bgen b$ is the circle (the
  classifying type of the free group on one generator, $\mathbb{Z}$) and
  $\Bgen{a,c}$ is the classifying type of the free group on two generators.  We
  need to show that all maps in the span parts are $0$-truncated maps of
  $1$-types, because then the above theorem kicks in, showing in the end that
  $\BH$ is a $1$-type, with all four elements $a,b,c,d$ generating infinite
  cyclic subgroups of $\HG = \pi_1(\BH)$.
  Indeed, if we can show that we have a span of \(0\)-truncated maps between
  \(1\)-types in the right square in \eqref{eq:higman-decomp}, then
  \cref{thm:warn} tells us that \(\BH\) is a \(1\)-type and that, e.g., the map
  \(\Bgen{a,b,c} \to \BH\) is \(0\)-truncated.
  Moreover, if we can then show that we have a span of \(0\)-truncated maps
  between \(1\)-types in the left square in \eqref{eq:higman-decomp}, then
  \cref{thm:warn} implies that the map
  \(\Bgen b \to \Bgen{b,c} \to \Bgen{a,b,c}\) is also \(0\)-truncated.
  Combining this with the above, we see that the composite map
  \(\Bgen b \to \BH\) is 0-truncated, i.e., that the generator \(b\) has
  infinite order in \(\HG\).
  One can similarly derive this for the other generators.

  We start by looking at the types of the form $\Bgen{a,b}$.  These are the
  classifying types of the Baumslag--Solitar~\cite{BS1962} groups
  \[
    \mathrm{BS}(1,2) = \langle\, a,b \mid aba^{-1}=b^2\,\rangle,
  \]
  and are so-called HNN-extensions~\cite{HNN1949,LyndonSchupp2001}, so we have
  coequalizer diagrams,%
  \footnote{%
  In ordinary \(1\)-category theory, every coequalizer is epic. The usual proof
  shows that every (homotopy) coequalizer is \(0\)-epic.
  But coequalizers need not be \(1\)-epic (even coequalizers of \(1\)-types):
  the coequalizer \(\Bgen{b} \to \Bgen{a,b}\) is \(0\)-truncated but not
  \(1\)-epic as this would make it \(0\)-connected
  (by~\cref{1-acyclic-characterization}) and hence an equivalence, which it is
  not.}
  as below left, or equivalently pushouts, as
  below right:
  \[
    \begin{tikzcd}
      \Circle \ar[r,shift left=2,"b"]\ar[r,shift right=2,"b^2"'] & \Bgen b \ar[r] &
      \Bgen{a,b}
    \end{tikzcd}
    \quad
    \begin{tikzcd}
      \Circle+\Circle \ar[r,"\nabla"]\ar[d,"{[1,2]}"']\pocorner & \Circle\ar[d] \\
      \Bgen b \ar[r] & \Bgen{a,b}
    \end{tikzcd}
  \]
  The maps in the span of the pushout square are $0$-truncated maps of
  $1$-types, so~\cref{thm:warn} applies.  The (identical) inclusion maps can be
  identified with the map $\Bgen b \to \Bgen{a,b}$.  The other inclusion,
  $\Bgen a \to \Bgen{a,b}$, is also $0$-truncated, as it has a retraction,
  $\Bgen{a,b} \to \Bgen a$, defined by sending $b$ to the neutral element; this
  is well defined, since the relation becomes $aa^{-1}=1$.

  This takes care of the input span to the left pushout square
  in~\eqref{eq:higman-decomp}.  It remains to see that the maps of the form
  $\Bgen{a,c}\to\Bgen{a,b,c}$ are $0$-truncated.  The follows by
  descent~\cite[Sec.~25]{Rijke2022} from looking at the commutative cube:
  \[
    \begin{tikzcd}
      & \One\ar[dl]\ar[d]\ar[dr] & \\
      \Bgen a\ar[d] & \Bgen b\ar[dl]\ar[dr] & \Bgen c\ar[d] \\
      \Bgen{a,b}\ar[dr] & \Bgen{a,c}\ar[d]\ar[from=2-1,crossing over]\ar[from=2-3,crossing over] & \Bgen{b,c}\ar[dl] \\
      & \Bgen{a,b,c} &
    \end{tikzcd}
  \]
  The back faces are pullbacks: The subgroup $\gen b$ is normal in $\gen{a,b}$,
  and has trivial intersection with the subgroup $\gen a$, so the pullback of
  $\Bgen a\to \Bgen{a,b} \leftarrow \Bgen b$ is the double coset
  $\gen a\backslash \gen{a,b} / \gen b$. Since the product of the two subgroups
  is the whole group, this is contractible, and similarly for the back right
  face.  In addition, the top and bottom faces are pushouts, so the front faces
  are pullbacks as well by descent.  Since the front bottom maps are
  (individually and jointly) surjective, and the maps on the sides are
  $0$-truncated, the map in front is as well, as desired.
\end{proof}
We conclude that $\BH$ is a nontrivial acyclic $1$-type.

There are also analogs of the Higman group with any number $n$ of generators,
and the same argument shows that these classify infinite acyclic groups for
$n\ge4$.  For $n<4$, these groups are trivial, see e.g.~\cite{Samuel2020}
for the case $n=3$.

The usual proofs that the Higman group is not the trivial group,
e.g.~\cite[Prop.~6(b), Sec.~1.4]{Serre1980}, rely on nontrivial results from
combinatorial group theory, specifically that we have embeddings into HNN
extensions and amalgamations~\cite[Thms.~2.1(I) and 2.6, Ch.~IV]{LyndonSchupp2001}.
A noteworthy aspect of our proof is that it completely avoids combinatorial
group theory and associated case distinctions, instead using tools from higher
topos theory such as flattening (cf.~\cref{flattening-avoids-excluded-middle}).

Finally, we note that nullification at any nontrivial acyclic type, such as
$\BH$ or Hatcher's example \(X\) of \cref{sec:hatcher}, provides a nontrivial
modality for which all types are separated, as conjectured
in~\cite[Ex.~6.6]{ChristensenRijke2022}.

\section{Concluding remarks}\label{sec:conclusion}

In this paper we characterized the epimorphisms in univalent mathematics as the
acyclic maps.
The ensuing study of acyclic types, and the relativized versions to \(k\)-types,
led to a further development of synthetic homotopy theory with applications in
group theory.

There are numerous directions for future work. Our primary objective is to
establish the acyclic maps as a modality~\cite{RSS2020}, as Raptis and Hoyois
did in the context of higher topos theory~\cite{Raptis2019,Hoyois2019}.
Moreover, we would hope to show that this modality is accessible, perhaps
under a mild extra assumption. In spaces, it can be explicitly described
as nullification at a small collection of spaces~\cite{BerrickCasacuberta1999}.
An essential ingredient in constructing the modality is Quillen's
plus-construction~\cite[Def.~1.4.1]{Kbook}.

We left open the question of giving a cohomological characterization of \(k\)-acyclicity,
as well as whether we can prove in HoTT that cohomologically acyclic types
are homologically acyclic.

Another thread for future research is to construct acyclic types of a
different nature than the examples presented in this paper, by considering
automorphism groups~\cite{HarpeMcDuff1983}, e.g. \(\Aut(\mathbb N)\), or binate
groups~\cite{Berrick1989}.

Additionally, we could work towards type-theoretic developments of the
Barratt--Priddy(--Quillen) theorem~\cite{BarrattPriddy} and the Kan--Thurston
theorem~\cite{KanThurston1976}. The latter says that an
\(\infty\)\nobreakdash-group can be presented by a pair \((G,P)\) of a group
\(G\) and perfect normal subgroup \(P \triangleleft G\)~\cite{BDH1980}.

Finally, we could try to generalize our results on epimorphisms to arbitrary
wild categories with pullbacks and universal pushouts that satisfy descent. We
might also have to impose the requirement that these categories are locally
cartesian closed for the notion of a dependent epimorphism to make sense.

\section*{Acknowledgements}
We are grateful to Mathieu Anel for~\cref{Whitehead-plus-principles} and to
Fredrik Nordvall Forsberg for detailed comments that helped to improve the
paper.
Further, we thank the anonymous reviewers for their thorough reports and
insightful comments. In particular they drew our attention to an omission in the
proof of \cref{epi-of-groups-characterization} and suggested the current proof
of \cref{acyclic-iff-balanced} (our previous proof relied on the plus principle
in one direction).

\section*{Funding}
The second author was supported by The Royal Society (grant reference
URF\textbackslash R1\textbackslash\\191055).
The third author was supported by the TydiForm project and the MURI grant, US
Air Force Office of Scientific Research, award numbers FA9550-21-1-0024 and
FA9550-21-1-0009 respectively.

\setcounter{biburllcpenalty}{7000}
\setcounter{biburlucpenalty}{8000}
\printbibliography

@book{AA2004,
  title =	 {Abstract Algebra},
  author =	 {Dummit, David S. and Foote, Richard M.},
  publisher =	 {John Wiley \& Sons, Inc.},
  year =	 2004,
  edition =	 {Third}
}

@Article{         ABFJ2020,
  Author =	 {Anel, Mathieu and Biedermann, Georg and Finster, Eric and
                  Joyal, Andr\'{e}},
  Title =	 {A generalized {B}lakers-{M}assey theorem},
  Journal =	 {Journal of Topology},
  Volume =	 13,
  Year =	 2020,
  Number =	 4,
  Pages =	 {1521--1553},
  DOI =		 {10.1112/topo.12163}
}

@Unpublished{     BCFR,
  Author =	 {Buchholtz, Ulrik and Christensen, J. Daniel and Flaten, Jarl
                  G. Taxerås and Rijke, Egbert},
  Title =	 {Central {H}-spaces and banded types},
  ArchivePrefix ={arXiv},
  EPrint =	 {2301.02636},
  Year =	 2023
}

@software{Agda,
  author = {Norell, Ulf and Danielsson, Nils Anders and Cockx, Jesper and Abel, Andreas and {others}},
  title = {Agda},
  url = {https://wiki.portal.chalmers.se/agda/pmwiki.php}
}

@article{BDH1980,
  title =	 {The topology of discrete groups},
  journal =	 {Journal of Pure and Applied Algebra},
  volume =	 16,
  number =	 1,
  pages =	 {1--47},
  year =	 1980,
  author =	 {Gilbert Baumslag and Eldon Dyer and Alex Heller},
  doi =		 {10.1016/0022-4049(80)90040-7}
}

@InProceedings{BDR2018,
  author =	 {Buchholtz, Ulrik and van Doorn, Floris and Rijke, Egbert},
  title =	 {Higher Groups in Homotopy Type Theory},
  year =	 2018,
  publisher =	 {Association for Computing Machinery},
  doi =		 {10.1145/3209108.3209150},
  booktitle =	 {Proceedings of the 33rd Annual ACM/IEEE Symposium on Logic in
                  Computer Science},
  pages =	 {205--214},
  series =	 {LICS ’18},
}

@Article{         BS1962,
  Author =	 {Baumslag, Gilbert and Solitar, Donald},
  Title =	 {Some two-generator one-relator non-{H}opfian groups},
  Journal =	 {Bulletin of the American Mathematical Society},
  Volume =	 68,
  Year =	 1962,
  Pages =	 {199--201},
  DOI =		 {10.1090/S0002-9904-1962-10745-9}
}

@Article{         BarrattPriddy,
  Author =	 {Barratt, Michael and Priddy, Stewart},
  Title =	 {On the homology of non-connected monoids and their associated
                  groups},
  Journal =	 {Commentarii Mathematici Helvetici},
  Volume =	 47,
  Year =	 1972,
  Pages =	 {1--14},
  DOI =		 {10.1007/BF02566785}
}

@InProceedings{Berrick1989,
  author =	 {Berrick, A. J.},
  title =	 {Universal groups, binate groups and acyclicity},
  booktitle =	 {Group Theory},
  booksubtitle = {Proceedings of the Singapore Group Theory Conference held at
                  the National University of Singapore, June 8--19, 1987},
  editor =	 {Cheng, Kai N. and Leong, Yu K.},
  year =	 1989,
  series =	 {De Gruyter Proceedings in Mathematics},
  publisher =	 {Walter de Gruyter},
  doi =		 {10.1515/9783110848397-019}
}

@Article{	  BerrickCasacuberta1999,
  author =	 {Berrick, A. J. and Casacuberta, Carles},
  title =	 {A universal space for plus-constructions},
  journal =	 {Topology},
  volume =	 38,
  year =	 1999,
  number =	 3,
  pages =	 {467--477},
  doi =		 {10.1016/S0040-9383(97)00073-6}
}

@article{CORS2020,
  title =	 {Localization in homotopy type theory},
  author =	 {Christensen, J.~Daniel and Opie, Morgan and Rijke, Egbert and
                  Scoccola, Luis},
  year =	 2020,
  volume =	 4,
  number =	 1,
  pages =	 {1--32},
  journal =	 {Higher Structures},
  doi =		 {10.21136/HS.2020.01}
}

@article{ChristensenRijke2022,
  title =	 {Characterizations of modalities and lex modalities},
  author =	 {Christensen, J.~Daniel and Rijke, Egbert},
  journal =	 {Journal of Pure and Applied Algebra},
  volume =	 226,
  number =	 3,
  year =	 2022,
  doi =		 {10.1016/j.jpaa.2021.106848}
}

@article{ChristensenScoccola2023,
  title =	 {The {Hurewicz} theorem in homotopy type theory},
  author =	 {Christensen, J.~Daniel and Scoccola, Luis},
  journal =	 {Algebraic \& Geometric Topology},
  volume =	 23,
  number =	 5,
  year =	 2023,
  pages =	 {2107--2140},
  doi =		 {10.1016/j.jpaa.2021.106848}
}

@article{DyerVasquez1973,
  title =	 {Some small aspherical spaces},
  volume =	 16,
  number =	 3,
  journal =	 {Journal of the Australian Mathematical Society},
  publisher =	 {Cambridge University Press},
  author =	 {Dyer, Eldon and Vasquez, Alphonse T.},
  year =	 1973,
  pages =	 {332--352},
  doi =		 {10.1017/S1446788700015147}
}

@unpublished{     EscardoFreeGroup,
  author =	 {Marc Bezem and Thierry Coquand and Peter Dybjer and Mart\'in
                  Escard\'o},
  title =	 {Free groups in {HoTT/UF} in {Agda}},
  year =	 2021,
  howpublished =
                  {\url{https://www.cs.bham.ac.uk/~mhe/TypeTopology/FreeGroup.html}},
}

@article{HNN1949,
  author =	 {Higman, Graham and Neumann, Bernhard H. and Neumann, Hanna},
  title =	 {Embedding Theorems for Groups},
  journal =	 {Journal of the London Mathematical Society},
  volume =	 {s1-24},
  number =	 4,
  pages =	 {247--254},
  doi =		 {10.1112/jlms/s1-24.4.247},
  year =	 1949
}

@Article{	  HarpeMcDuff1983,
  author =	 {de la Harpe, Pierre and McDuff, Dusa},
  title =	 {Acyclic groups of automorphisms},
  journal =	 {Commentarii Mathematici Helvetici},
  volume =	 58,
  year =	 1983,
  number =	 1,
  pages =	 {48--71},
  doi =		 {10.1007/BF02564624}
}

@Book{            HatcherAT,
  Author =	 {Hatcher, Allen},
  Title =	 {Algebraic topology},
  Publisher =	 {Cambridge University Press},
  Year =	 2002,
  URL =		 {https://pi.math.cornell.edu/~hatcher/AT/ATpage.html}
}

@article{HausmannHusemoller1979,
  title =	 {Acyclic maps},
  author =	 {Jean-Claude Hausmann and Dale Husem\"oller},
  year =	 1979,
  volume =	 25,
  number =	 {1--2},
  journal =	 {L'enseignement Math\'ematique},
  pages =	 {53--75},
  doi =		 {10.5169/seals-50372}
}

@Article{	  Higman1951,
  author =	 {Higman, Graham},
  title =	 {A finitely generated infinite simple group},
  journal =	 {The Journal of the London Mathematical Society},
  volume =	 26,
  year =	 1951,
  pages =	 {61--64},
  doi =		 {10.1112/jlms/s1-26.1.61}
}

@Book{HoTTBook,
  author =	 {The {Univalent Foundations Program}},
  title =	 {Homotopy Type Theory: Univalent Foundations of Mathematics},
  publisher =	 {\url{https://homotopytypetheory.org/book}},
  address =	 {Institute for Advanced Study},
  year =	 2013
}

@Unpublished{     Hoyois2019,
  Title =	 {On {Q}uillen's plus construction},
  Author =	 {Marc Hoyois},
  URL =		 {https://hoyois.app.uni-regensburg.de/papers/acyclic.pdf},
  Year =	 2019
}

@Article{	  KanThurston1976,
  author =	 {Kan, Daniel M. and Thurston, William P.},
  title =	 {Every connected space has the homology of a {$K(\pi ,1)$}},
  journal =	 {Topology},
  volume =	 15,
  year =	 1976,
  number =	 3,
  pages =	 {253--258},
  doi =		 {10.1016/0040-9383(76)90040-9}
}

@Book{            Kbook,
  Author =	 {Weibel, Charles A.},
  Title =	 {The {$K$}-book},
  Series =	 {Graduate Studies in Mathematics},
  Volume =	 145,
  Subtitle =	 {An introduction to algebraic $K$-theory},
  Publisher =	 {American Mathematical Society},
  Year =	 2013,
  DOI =		 {10.1090/gsm/145},
}

@Book{            LurieHTT,
  Author =	 {Lurie, Jacob},
  Title =	 {Higher topos theory},
  Series =	 {Annals of Mathematics Studies},
  Volume =	 170,
  Publisher =	 {Princeton University Press},
  Year =	 2009,
  DOI =		 {10.1515/9781400830558},
  URL =		 {https://www.math.ias.edu/~lurie/papers/HTT.pdf}
}

@book{LyndonSchupp2001,
  author =	 {Lyndon, Roger C. and Schupp, Paul E.},
  title =	 {Combinatorial Group Theory},
  year =	 2001,
  note =	 {Reprint of the 1977 Edition (Ergebnisse der Mathematik und
                  ihrer Grenzgebiete, Vol. 89)},
  publisher =	 {Springer},
  series =	 {Classics in Mathematics},
  doi =		 {10.1007/978-3-642-61896-3}
}

@book{MinesRichmanRuitenburg1988,
  author =	 {Ray Mines and Fred Richman and Wim Ruitenburg},
  title =	 {A Course in Constructive Algebra},
  year =	 1988,
  series =	 {Universitext},
  publisher =	 {Springer},
  doi =		 {10.1007/978-1-4419-8640-5}
}

@article{RSS2020,
  title =	 {Modalities in homotopy type theory},
  author =	 {Egbert Rijke and Michael Shulman and Bas Spitters},
  doi =		 {10.23638/LMCS-16(1:2)2020},
  journal =	 {Logical Methods in Computer Science},
  volume =	 16,
  number =	 1,
  year =	 2020,
}

@Article{         Raptis2019,
  Author =	 {Raptis, George},
  Title =	 {Some characterizations of acyclic maps},
  Journal =	 {Journal of Homotopy and Related Structures},
  Volume =	 14,
  Year =	 2019,
  Number =	 3,
  Pages =	 {773--785},
  DOI =		 {10.1007/s40062-019-00231-6},
}

@unpublished{Rezk2019,
  title =	 {{Lectures on Higher Topos Theory (Leeds, June 2019)}},
  year =	 2019,
  author =	 {Rezk, Charles},
  url =		 {https://rezk.web.illinois.edu/leeds-lectures-2019.pdf}
}

@unpublished{Rijke2022,
  title =	 {Introduction to Homotopy Type Theory},
  author =	 {Egbert Rijke},
  year =	 2022,
  note =	 {Book draft; version of 8 April},
  url =
                  {https://raw.githubusercontent.com/martinescardo/HoTTEST-Summer-School/main/HoTT/hott-intro.pdf}
}

@Misc{            Samuel2020,
  Title =	 {{Let $G$ be a group. Let $x,y,z \in G$ such that $[x,y]=y$,
                  $[y,z]=z$, $[z,x]=x$. Prove that $x=y=z=e$.}},
  Author =	 {Matt Samuel},
  HowPublished = {Mathematics Stack Exchange},
  Date =	 {2020-06-12},
  URL =		 {https://math.stackexchange.com/q/3717004}
}

@software{agda-unimath,
author = {Rijke, Egbert and Bonnevier, Elisabeth and Prieto-Cubides, Jonathan and Bakke, Fredrik and {others}},
license = {MIT},
title = {{The agda-unimath library}},
url = {https://unimath.github.io/agda-unimath/}
}

@Article{         Scoccola2020,
  Author =	 {Scoccola, Luis},
  Title =	 {Nilpotent types and fracture squares in homotopy type theory},
  Journal =	 {Mathematical Structures in Computer Science},
  Volume =	 30,
  Year =	 2020,
  Number =	 5,
  Pages =	 {511--544},
  DOI =		 {10.1017/s0960129520000146}
}

@book{Serre1980,
  title =	 {Trees},
  DOI =		 {10.1007/978-3-642-61856-7},
  publisher =	 {Springer-Verlag},
  author =	 {Serre, Jean-Pierre},
  year =	 1980
}

@unpublished{Shulman2019,
  title =	 {All \((\infty,1)\)-toposes have strict univalent universes},
  author =	 {Shulman, Michael},
  year =	 2019,
  month =	 4,
  eprint =	 {1904.07004},
  archivePrefix ={arXiv}
}

@misc{            Symmetry,
  title =	 {Symmetry},
  author =	 {Bezem, Marc and Buchholtz, Ulrik and Cagne, Pierre and Dundas,
                  Bj{\o}rn Ian and Grayson, Daniel R.},
  date =	 {2024-04-06},
  howpublished = {\url{https://github.com/UniMath/SymmetryBook}},
  note =	 {Commit: \texttt{994b4f1}}
}

@unpublished{Warn2023,
  Author =	 {Wärn, David},
  Title =	 {Path spaces of pushouts},
  Year =	 2023,
  URL =		 {https://dwarn.se/po-paths.pdf}
}

@book{cwm,
  Author =	 {Mac Lane, Saunders},
  Title =	 {Categories for the Working Mathematician},
  Edition =	 2,
  Publisher =	 {Springer},
  Year =	 1978,
  Series =	 {Graduate Texts in Mathematics},
  Volume =	 5,
  DOI =		 {10.1007/978-1-4757-4721-8}
}

@unpublished{Ljungstrom2024,
  author =	 {Ljungstr\"om, Axel},
  title =	 {Symmetric Monoidal Smash Products in Homotopy Type Theory},
  archivePrefix ={arXiv},
  eprint =	 {2402.03523},
  year =	 2024,
}

@phdthesis{Brunerie2016,
  title =	 {On the homotopy groups of spheres in homotopy type theory},
  author =	 {Brunerie, Guillaume},
  year =	 2016,
  school =	 {Universit\'e de Nice Sophia Antipolis --- UFR Sciences},
  archivePrefix ={arXiv},
  eprint =	 {1606.05916}
}

@phdthesis{vanDoorn2018,
  title =	 {On the Formalization of Higher Inductive Types and Synthetic
                  Homotopy Theory},
  author =	 {van Doorn, Floris},
  year =	 2018,
  school =	 {Carnegie Mellon University},
  archivePrefix ={arXiv},
  eprint =	 {1808.10690}
}

@InProceedings{CagneBuchholtzKrausBezem2024,
  Author =	 {Cagne, Pierre and Buchholtz, Ulrik and Kraus, Nicolai and
                  Bezem, Marc},
  Title =	 {On symmetries of spheres in univalent foundations},
  Year =	 2024,
  Publisher =	 {Association for Computing Machinery},
  Address =	 {New York, NY, USA},
  DOI =		 {10.1145/3661814.3662115},
  BookTitle =	 {Proceedings of the 39th Annual ACM/IEEE Symposium on Logic in
                  Computer Science},
  articleno =	 20,
  numpages =	 14,
  Location =	 {Tallinn, Estonia},
  Series =	 {LICS '24},
  archivePrefix ={arXiv},
  eprint =	 {2401.15037}
}

@Misc{            UnitFreeAbelianGroup,
  Title =	 {{Constructively, is the unit of the ``free abelian group''
                  monad on sets injective?}},
  Author =	 {W\"arn, David},
  HowPublished = {MathOverflow answer},
  Date =	 {2022-10-02},
  URL =		 {https://mathoverflow.net/a/431619}
}

@Article{         TaxerasFlaten2022,
  Title =	 {Univalent categories of modules},
  Volume =	 33,
  DOI =		 {10.1017/S0960129523000178},
  Number =	 2,
  Journal =	 {Mathematical Structures in Computer Science},
  Author =	 {Flaten, Jarl G. Taxerås},
  Year =	 2023,
  Pages =	 {106–133},
  archivePrefix ={arXiv},
  eprint =	 {2207.03261}
}

@Misc{		  ChristensenFlatenExt,
  title =	 {Ext groups in Homotopy Type Theory},
  author =	 {J. Daniel Christensen and Jarl G. Taxerås Flaten},
  year =	 2023,
  eprint =	 {2305.09639},
  archiveprefix ={arXiv},
  primaryclass = {math.AT}
}

@article{Harting1982,
  title =	 {Internal coproduct of abelian groups in an elementary topos},
  volume =	 10,
  number =	 11,
  journal =	 {Communications in Algebra},
  author =	 {Harting, Roswitha},
  year =	 1982,
  DOI =		 {10.1080/00927878208822773},
  pages =	 {1173--1237}
}

@book{Riehl2016,
  title =	 {Category Theory in Context},
  author =	 {Riehl, Emily},
  year =	 2016,
  publisher =	 {Dover Publications},
  series =	 {Aurora: Dover Modern Math Originals}
}

@book{Borceux1994,
  title =	 {Handbook of Categorical Algebra},
  subtitle =	 {Basic Category Theory},
  year =	 1994,
  volume =	 1,
  author =	 {Borceux, Francis},
  number =	 50,
  series =	 {Encyclopedia of Mathematics and its Applications},
  doi =		 {10.1017/CBO9780511525858}
}

@unpublished{Trimble2020,
  title =	 {epimorphisms of groups are surjective},
  author =	 {Trimble, Todd},
  url =
                  {https://ncatlab.org/nlab/show/epimorphisms+of+groups+are+surjective},
  year =	 2020,
}

@article{Alonso1983,
  title =	 {Fibrations that are cofibrations},
  volume =	 87,
  DOI =		 {10.1090/s0002-9939-1983-0687656-3},
  number =	 4,
  journal =	 {Proceedings of the American Mathematical Society},
  author =	 {Alonso, Juan M.},
  year =	 1983,
  pages =	 {749--753}
}

@InBook{	  Rathjen2018,
  title =	 {Proof Theory of Constructive Systems: Inductive Types and
                  Univalence},
  author =	 {Michael Rathjen},
  year =	 2018,
  booktitle =	 {Feferman on Foundations},
  booksubtitle = {Logic, Mathematics, Philosophy},
  editor =	 {Gerhard J\"ager and Wilfried Sieg},
  series =	 {Outstanding Contributions to Logic},
  volume =	 13,
  pages =	 {385--419},
  doi =		 {10.1007/978-3-319-63334-3_14}
}

@InBook{	  Buchholtz2019,
  author =	 "Buchholtz, Ulrik",
  editor =	 "Centrone, Stefania and Kant, Deborah and Sarikaya, Deniz",
  title =	 "Higher Structures in Homotopy Type Theory",
  booktitle =	 "Reflections on the Foundations of Mathematics",
  booksubtitle = "Univalent Foundations, Set Theory and General Thoughts",
  year =	 2019,
  publisher =	 "Springer",
  series =	 "Synthese Library",
  volume =	 407,
  pages =	 "151--172",
  doi =		 "10.1007/978-3-030-15655-8_7"
}

@InBook{	  Shulman2021,
  title =	 {Homotopy Type Theory: The Logic of Space},
  booktitle =	 {New Spaces in Mathematics},
  booksubtitle = {Formal and Conceptual Reflections},
  publisher =	 {Cambridge University Press},
  author =	 {Shulman, Michael},
  editor =	 {Anel, Mathieu and Catren, Gabriel},
  year =	 2021,
  pages =	 {322--404},
  doi =		 {10.1017/9781108854429.009}
}

@InBook{	  Awodey2012,
  author =	 "Awodey, Steve",
  editor =	 "Dybjer, Peter and Lindstr{\"o}m, Sten and Palmgren, Erik and
                  Sundholm, G{\"o}ran",
  title =	 "Type Theory and Homotopy",
  booktitle =	 "Epistemology versus Ontology",
  booksubtitle = "Essays on the Philosophy and Foundations of Mathematics in
                  Honour of Per Martin-L{\"o}f",
  series =	 "Logic, Epistemology, and the Unity of Science",
  volume =	 27,
  year =	 2012,
  publisher =	 "Springer",
  pages =	 "183--201",
  doi =		 "10.1007/978-94-007-4435-6_9"
}

@InCollection{	  Adams1969,
  author =	 {Adams, J. Frank},
  title =	 {Lectures on generalised cohomology},
  booktitle =	 {{Category Theory, Homology Theory and their Applications III}},
  series =	 {Lecture Notes in Mathematics},
  number =	 99,
  pages =	 {1--138},
  publisher =	 {Springer-Verlag},
  editor =	 {Hilton, Peter J.},
  year =	 1969,
  doi =		 {10.1007/BFb0081960}
}

@InProceedings{BLM2022,
  DOI =		 {10.4230/LIPICS.CSL.2022.11},
  Author =	 {Brunerie, Guillaume and Ljungstr\"{o}m, Axel and M\"{o}rtberg,
                  Anders},
  Title =	 {Synthetic Integral Cohomology in Cubical Agda},
  Publisher =	 {Schloss Dagstuhl – Leibniz-Zentrum f\"{u}r Informatik},
  Year =	 2022,
  series =	 {Leibniz International Proceedings in Informatics (LIPIcs)},
  volume =	 216,
  editor =	 {Manea, Florin and Simpson, Alex},
  booktitle =	 {30th EACSL Annual Conference on Computer Science Logic (CSL
                  2022)},
}

@mastersthesis{Cavallo2015,
  title =	 {Synthetic cohomology in homotopy type theory},
  author =	 {Cavallo, Evan},
  year =	 2015,
  school =	 {Carnegie Mellon University},
  type =	 {Master's thesis},
  url =		 {https://ecavallo.net/works/thesis15.pdf}
}

\end{document}